\documentclass[a4paper,12pt]{article}
\usepackage{latexsym,amssymb,amsfonts,amsmath,amsthm}
\usepackage{graphicx}

\newtheorem{definition}{Definition}[section]
\newtheorem{theorem}[definition]{Theorem}
\newtheorem{lemma}[definition]{Lemma}

\newtheorem{proposition}[definition]{Proposition}
\newtheorem{remark}[definition]{Remark}
\newtheorem{problem}[definition]{Problem}
\newtheorem*{main}{Main Result}

\newcommand{\bs}[1]{\boldsymbol{#1}}
\newcommand{\im}{\bs{\rm i}}

\newcommand{\spann}{{\rm span}}

\newcommand{\bra}{\langle}
\newcommand{\ket}{\rangle}

\numberwithin{equation}{section}

\title{{\Large {\bf Quantum Search on Simplicial Complexes
}
}}
\author{ 
Kaname Matsue$^{1}$ 
\footnote{ 
{\tt kmatsue@imi.kyushu-u.ac.jp} 
},\quad
Osamu Ogurisu$^{2}$ 
\footnote{
{\tt ogurisu@staff.kanazawa-u.ac.jp}  
}\quad
and
Etsuo Segawa$^{3}$ 
\footnote{
{\tt e-segawa@m.tohoku.ac.jp}
}
\\
{\scriptsize $^1$ 
Institute of Mathematics for Industry/International Institute for Carbon-Neutral Energy Research (WPI-I$^2$CNER), Kyushu University, 
}\\
{\scriptsize 
Fukuoka 819-0395, Japan 
} \\
{\scriptsize $^2$ 
Division of Mathematical and Physical Sciences, Kanazawa University 
}\\
{\scriptsize 
Kanazawa, Ishikawa 920-1192, Japan
} \\
{\scriptsize $^3$ 
Graduate School of Information Sciences, Tohoku University,
}\\
{\scriptsize 
Aoba, Sendai 980-8579, Japan
} \\
}
\date{\empty }
\begin{document}
\maketitle

\par\noindent
\begin{small}
\par\noindent
{\bf Abstract}. 
In this paper, we propose an extension of quantum searches on graphs driven by quantum walks to simplicial complexes. 
To this end, we define a new quantum walk on simplicial complex which is an alternative of preceding studies by authors. 
We show that the quantum search on the specific simplicial complex corresponding to the triangulation of $n$-dimensional unit square 
driven by this new simplicial quantum walk works well, 
namely, a marked simplex can be found with probability $1+o(1)$ within a time $O(\sqrt{N})$, where $N$ is the number of simplices with the dimension of marked simplex.


\end{small}

\setcounter{equation}{0}

%
%
\section{Introduction}
The quantum walk is a quantum analogue of classical random walks \cite{Gu}. 
Its primitive form of the discrete-time quantum walk on $\mathbb{Z}$ can be seen in Feynman's checker board \cite{FH}. 
It is mathematically shown (e.g. \cite{K1}) from a combinatorial and probabilistic approach 
that this quantum walk has a completely different limiting behavior from random walks, 
which is a typical example showing a difficulty of intuitive description of quantum walks' behavior from a classical process.
By such an interesting observation and also the efficiency of quantum walks in quantum search algorithms (see \cite{A1, K2} and their references), 
quantum walks are studied from various kinds of viewpoints such as not only the quantum information and mathematics, but also 
a quantum simulation of physical process derived from the Dirac and Schr\"{o}dinger equations, experimental and engineering viewpoints, and so on.
\par

The time evolution of the discrete-time quantum walk is given by discrete iterations of a unitary operator on $\ell^2$-summable Hilbert space generated by arcs of a given graph. 
The unitary operator $U$ is determined by local unitary operators assigned at all vertices. 
Let $\psi_n$ be the $n$-th iteration of the discrete-time, that is, $\psi_n=U\psi_{n-1}$. 
Due to the unitarity of the time evolution, we can obtain a probability distribution $\mu_n$ from each time iteration of this walk; 
that is, we can define a map $\psi_n\mapsto \mu_n$. 
We call this map a measurement. We are interested in the sequence of $\{\mu_n\}_n$.

One of the interesting research direction is to explore how topological features of underlying objects affect asymptotics of the probability distribution $\mu_n$. 
For example, in infinite graph cases, it is shown that a homological structure~\cite{HKSS2} and also existence of finite energy flow~\cite{HS} of graphs provide localization of the Grover walk. 
An infinite abelian covering provides the linear spreading which is quadratically faster than diffusive spreading (e.g., \cite{HKSS2}). 
A sensitivity of quantum walks to boundaries of graphs is one of the main stream of quantum walk's study from the viewpoint of topological phases, for example~\cite{AE, CGSVWW, KRBD, OK}.  
For finite graph cases, estimations of effectiveness of quantum searches on graphs are one of the main topics~\cite{ADFP2012, ADMP2010, P2013, S2008, SKW2003}.
For example, 
search algorithms of the target vertices are considered on several graphs such as finite $d$-dimensional grid~\cite{AKR}, hypercubes \cite{SKW2003}, honeycomb network \cite{ADMP2010}, triangular lattice \cite{ADFP2012} and other graphs like Johnson graphs \cite{A2}.
To extract the graph structures which accomplish the quantum speed up and also perfect state transfer~\cite{SS} is one of the interesting inverse problem. 
As another interesting property of quantum walk, a graph centrality induced by quantum walks is also proposed by~\cite{Jingbo}. 
A classification of graphs from the viewpoint of the periodicity of quantum walks also has recently proceeded~\cite{HKSS1,Yoshie}.

%
\par
We believe that quantum walks can be defined on objects with mathematically richer structures than graphs in topological features, which is our main motivation of this study. 
Authors have introduced quantum walks on {\em simplicial complexes} (\cite{MOSver1}), which are higher-dimensional extension of graphs.
Based on the structure of Szegedy-type walks on graphs, coin operators and shift operators on simplicial complexes are introduced to define unitary operators on them, which are referred to as {\em simplicial quantum walks}. 
Unlike graphs, simplicial complexes admit not only connectivity and rings but also cavities, twists and more general topological features. 
In \cite{MOSver1}, numerical studies have shown that, keeping a fundamental feature such as ballistic spreading of walkers, simplicial quantum walks have responses to topology of simplicial complexes, such as localization and presence of nontrivial homology: algebraic description of holes in simplicial complexes, hierarchy of localizations with respect to the order of homology, and sensitivity of orientations on simplicial complexes. 
These features indicate that quantum walks have rich response with respect to topology of underlying geometric objects including graphs and simplicial complexes.
\par
\bigskip
Our aim here is to provide a quantum search algorithm over simplicial complexes in terms of simplicial quantum walks.
We believe that such an extension will build a bridge between quantum search procedures and various knowledge of topology and geometry, as well as a bridge between quantum walks and the latter.
As the starting point, we consider quantum search on the unit sphere $S^n$ (not embedded graphs, but the complex which is topologically identical to $S^n$), which are topologically different from Euclidean spaces and torus.
This is a good example of a series of studies since spheres are considered as ones of the simplest geometric objects.
On the other hand, simplicial quantum walks introduced in \cite{MOSver1} requires a lot of coin states on each simplex. 
In other words, as for simplicial quantum walks on an $n$-dimensional simplicial complex $\mathcal{K}$, the dimension of total spaces of quantum walks is proportional to $(n+1)!$, which will be costly for practical computations of quantum walks.
To make the situation simpler, we introduce an alternative version of simplicial quantum walks.
The new version pays attention to {\em orientations of simplices} and structure of unitary operators {\em based on the composite of Grover operators, rather than the composite of coin and shift operators}.
These features well match those of bipartite walks on bipartite graphs originated by Szegedy (e.g., \cite{Sze}) and coined quantum walks driven by Grover operators.
Indeed, we prove that the new version of simplicial quantum walks are unitary equivalent to these quantum walks on corresponding graphs, under suitable constraints on unitary operators and simplicial complexes (Section \ref{section-SQW}).
Such equivalence build bridges between simplicial quantum walks and quantum walks on graphs, including quantum search problems.
They also give a new insight of coined walks on special class of graphs from the viewpoint of quantum walk models on simplicial complexes.
\par
Our main result in this paper is the following whose details are discussed in successive sections.
\begin{main}[Theorem \ref{thm-search}]
For the $n(\geq 2)$-dimensional simplicial complex $\mathcal{K}$ as a triangulation\footnote{
Triangulation of a manifold or a surface means a simplicial complex $\mathcal{K}$ whose geometric realization; the union of all simplices in $\mathcal{K}$, is homeomorphic (topologically identical) to the manifold.
In the case of our statement, \lq\lq manifold" is a unit sphere $S^n$.
}
of the unit sphere $S^n$ (details are discussed in Section \ref{section-search}), fix an $(n-1)$-simplex $\tau_\ast \in \mathcal{K}$ as a marked simplex. 
Then the \lq\lq quantum search" driven by our new quantum walk on $\mathcal{K}$ finds $\tau_\ast$ within a time $t_f=O(n) = O(\sqrt{N})$ with probability $p_f\sim 1$, where $N$ is the number of $(n-1)$-simplices in $\mathcal{K}$. 
\end{main}
It turns out that the simplicial complex $\mathcal{K}$ possesses $O(N) = O(n^2)$ $(n-1)$-simplices, which immediately follows from the construction.
The essence of the result is therefore that the above \lq\lq quantum search" on $\mathcal{K}$ finds a marked simplex among $N$ entries with time complexity $O(\sqrt{N})$, which shows that our equipments achieve the quantum speed-up for search problems over simplicial complexes.
\par
\bigskip
The rest of this paper is organized as follows.
In Section \ref{section-SQW}, we define an alternative version of simplicial quantum walks discussed in \cite{MOSver1}\footnote{
Very recently, Luo and Tate \cite{LT2017} introduces an alternative form of quantum walks on simplicial complexes with a different motivation from our present study.
}.
The new version of quantum walks reflects information of orientations on simplices and reduces the number of states on them compared with the previous version in \cite{MOSver1}.
Moreover, these quantum walks turn out to be unitary equivalent to a class of quantum walks on graphs.
We also show that the equivalence of simplicial quantum walks on orientable simplicial complexes without boundary is realized by quantum walks on associated graphs with duplication structure, which simplifies the description of dynamics.
In particular, as the original simplicial quantum walks, our alternative walks also rely on geometry of simplicial complexes in terms of associated graph structures.
In Section \ref{section-search}, we consider the quantum search problem for simplicial quantum walks, which we shall call \lq\lq {\em simplicial quantum search}" problem, and prove the main result.
We also show quantum search with numerical simulations for demonstrating our main result in concrete situations.
In Appendix, fundamentals of simplicial complexes which require for our discussions, as well as the detailed proofs of spectral arguments in quantum search problems are collected.

%
%
\section{Simplicial quantum walks}
\label{section-SQW}

In this section, we define a quantum walk model on simplicial complexes, called a {\em simplicial quantum walk}.
This model is a higher dimensional analogue of quantum walks on graphs, and it is an alternative of the model which are defined in preceding work \cite{MOSver1}.
In this version, we pay attention to {\em orientability} on simplices, which reduces the number of states on each simplex and induces a well-defined quantum walk model.
Moreover, we show that the new model is unitary equivalent to a class of quantum walks on graphs.
The equivalence yields prospects of studies of simplicial quantum walks from the viewpoint of quantum walks on graphs including quantum search.
\par
The fundamental notions of simplicial complexes are listed in Appendix \ref{appendix-cpx} and hence readers who are not familiar with simplicial complexes can access basic knowledge and our requirements there.

\subsection{Setting and definition}
\label{section-setting}
Let $\mathcal{K}$ be an $n$-dimensional simplicial complex with $n\geq 2$.
We assume that $\mathcal{K}$ is strongly connected. 
Let $\mathcal{K}_{\kappa}$ be the collection of $\kappa$-dimensional simplices of $\mathcal{K}$ ($\kappa=0,\dots,n$). 
An element of $\mathcal{K}_{\kappa}$ is denoted by $|w_0w_1\dots w_\kappa|$. 
Here $|w_j|$ is a $0$-simplex, that is, a vertex. 
The order of $w_j$'s in $\mathcal{K}_{\kappa}$ is ignored, that is, $|w_0w_1\dots w_\kappa|=|w_{\pi^{-1}(0)}w_{\pi^{-1}(1)}\dots w_{\pi^{-1}(\kappa)}|$ 
for any $\pi\in S_{\kappa+1}$. Here $S_{\kappa+1}$ is the symmetric group on $\kappa+1$ letters. 
Define $\tilde {\mathcal{K}_\kappa}$ as the set of ordered sequences induced by $\mathcal{K}_\kappa$: $\tilde {\mathcal{K}_\kappa} := \{ (w_0,\dots,w_\kappa) \;|\; |w_0\dots w_\kappa| \in \mathcal{K}_{\kappa}\}$. 
Let \lq\lq $\sim_\kappa$\rq\rq  be the following equivalent relation on $\tilde {\mathcal{K}_\kappa}$: 
for two ordered sequences $\tilde w, \tilde w'\in \tilde {\mathcal{K}_\kappa}$, $\tilde w\sim_\kappa \tilde w'$ if and only if 
there exists an even permutation $\pi\in A_{\kappa+1}$ such that $\pi(\tilde w)\equiv \langle \tilde w_{\pi(0)} \cdots \tilde w_{\pi(n)}\rangle= \tilde w'$. 
Here $A_{\kappa+1}\subset S_{\kappa+1}$ is the alternating group on $\kappa+1$ letters. 
The quotient set $\langle \mathcal{K}_\kappa \rangle:= \tilde{\mathcal{K}_\kappa}/{\sim_\kappa}$ is denoted by 
$\langle \mathcal{K}_\kappa \rangle=\{\langle w_0w_1\dots w_{\kappa} \rangle \;|\; ( w_0,w_1,\dots, w_{\kappa})\in \tilde{\mathcal{K}_{\kappa}}\}$. 
\begin{definition}\rm
For $\sigma=\langle w_0\cdots w_{n} \rangle\in \langle \mathcal{K}_n\rangle$ and 
$\tau=\langle w_0'\cdots w_{n-1}' \rangle\in \langle \mathcal{K}_{n-1}\rangle$, we define 
$\sigma \triangleright \tau$ if and only if there exists $\pi\in A_{n+1}$ such that 
$\tau=\langle w_{\pi^{-1}(1)}\cdots w_{\pi^{-1}(n)}\rangle$. 
We call such $\tau$ an induced directed primary face of $\sigma$. 
\end{definition}
\begin{proposition}
\label{propwelldefine}
The induced directed primary face of $\sigma$ is well defined, 
that is, this is independent of the choice of representatives. 
\end{proposition}
\begin{proof}
For $n=1$, when $\sigma=\bra w_0w_1\ket$, the induced directed primary face is uniquely determined as $\tau=\bra w_1\ket$ since the even permutation
is only the identity operator. 
On the other hand, for $n\geq 2$, remark that for every $i\in \{0,1,\dots,n\}$, there exists $\pi\in A_{n+1}$ such that $\pi^{-1}(0)=i$. 
Let $\mu,\nu\in A_{n+1}$ be such that $\mu^{-1}(0)=\nu^{-1}(0)=i$.
It holds
$\langle w_{i}w_{\mu^{-1}(1)}\cdots w_{\mu^{-1}(n)}\rangle=\langle w_{i}w_{\nu^{-1}(1)}\cdots w_{\nu^{-1}(n)}\rangle$. 
We will show that even if we remove $w_i$ from both sides of the above, 
	$\langle w_{\mu^{-1}(1)}\cdots w_{\mu^{-1}(n)}\rangle=\langle w_{\nu^{-1}(1)}\cdots w_{\nu^{-1}(n)}\rangle$ holds. 
We set the cyclic permutation $\pi=(0,1,\dots,j)(j+1)\cdots (n)\in S_{n+1}$.
There exist $\mu', \nu' \in S_{n+1}$ such that $\mu=\mu'\circ \pi$ and $\nu=\nu'\circ \pi$. 
Therefore $\tau(w_{i}w_{\mu^{-1}(1)}\cdots w_{\mu^{-1}(n)})=(w_{i}w_{\nu^{-1}(1)}\cdots w_{\nu^{-1}(n)})$, 
where $\tau=\nu'\circ {\mu'}^{-1}$. 
If $i$ is even, then $\pi$ is even permutation, which implies $\mu'$ and $\nu'$ should be even permutation since $\mu=\mu'\circ\pi$ 
and $\nu=\nu'\circ \pi$ are even permutations. 
Thus $\tau=\nu'\circ \mu'$ is an even permutation. 
On the other hand, if $i$ is odd, then $\pi$, $\mu$ and $\nu$ are odd permutations, which implies $\tau$ is also an even permutation. 
This completes the proof. 
\end{proof}
We have shown that we can obtain all the induced directed primary faces $\tau\in \bra \mathcal{K}_{n-1} \ket$ of $\sigma\in \bra \mathcal{K}_n \ket$
by removing the initial letter of $\pi(\sigma)$ for arbitrary even permutation $\pi\in A_{n+1}$ ($n\geq 1$). 
Note that the set of induced directed primary faces of $\sigma=\bra w_0\dots w_n \ket$ ($n\geq 1$) is expressed by 
\begin{align}\label{iipp}
	\begin{cases} 
        \{\bra w_{j+1}\dots w_{n}w_0\dots w_{j-1} \ket \;|\; j=0,\dots,n \} & \text{: $n$ is odd,}  \\
        \{\bra w_{j+1}\dots w_{n}w_0\dots w_{j-1} \ket \;|\; j=0,\dots,n \} & \text{: $n$ is even, $j$=even,} \\
        \{\bra w_{j-1}\dots w_{0}w_n\dots w_{j+1} \ket \;|\; j=0,\dots,n \} & \text{: $n$ is even, $j$=odd.}
	\end{cases}
\end{align}
It is easy to check that the above permutation is an even permutation. 

\begin{definition}\rm
For $\sigma=\langle w_0\cdots w_{n} \rangle\in \langle \mathcal{K}_n\rangle$ and 
$\tau=\langle w_0'\cdots w_{n-1}' \rangle\in \langle \mathcal{K}_{n-1}\rangle$ $(n\geq 1)$, we define 
$\sigma \triangleright \tau$ with $\sigma=\bra w_0\cdots,w_n \ket$ if and only if there exists $j\in \{0,\dots,n\}$ such that 
$\tau$ is equivalent to a representative element in (\ref{iipp}). 
We call such $\tau$ an improved induced directed primary face of $\sigma$. 
\end{definition}
\begin{remark}
The relation  \lq\lq \;$\triangleright$" determining improved induced primary faces is well defined by Proposition~\ref{propwelldefine}. 
Note that the above definition itself still makes sense for $n=1$.
\end{remark}

The stage on which our quantum walker moves is constructed by $\mathcal{K}^{n,n-1}\subset \bra \mathcal{K}_n \ket\times \bra \mathcal{K}_{n-1} \ket$. 
Here
\begin{equation*}
\mathcal{K}^{n,n-1}=\{ (\sigma,\tau)\in \bra \mathcal{K}_n \ket\times \bra \mathcal{K}_{n-1} \ket \;|\; \sigma\triangleright \tau\}.
\end{equation*}
We provide two kinds of equivalence relations on $\mathcal{K}^{n,n-1}$: 
\begin{equation*}
(\sigma,\tau) \stackrel{\pi_1}{\sim} (\sigma',\tau') \ \stackrel{def}{\Leftrightarrow} \  \tau=\tau',\quad 
(\sigma,\tau) \stackrel{\pi_2}{\sim} (\sigma',\tau') \ \stackrel{def}{\Leftrightarrow} \  \sigma=\sigma'. 
\end{equation*}
The quotient sets are expressed by
\begin{equation*}
        \mathcal{K}^{n,n-1}/\stackrel{\pi_1}{\sim} = \{ E_\tau \;|\; \tau\in \bra \mathcal{K}_{n-1} \ket \}, \quad 
        \mathcal{K}^{n,n-1}/\stackrel{\pi_2}{\sim} = \{ F_\sigma \;|\; \sigma\in \bra \mathcal{K}_{n} \ket \},
\end{equation*}	
where 
\begin{equation}
\label{set-E-F}
E_\tau = \{ (\sigma,\tau) \;|\; \forall \sigma,\;\sigma\triangleright \tau \}, \quad
F_\sigma = \{ (\sigma,\tau) \;|\; \forall \tau,\;\sigma\triangleright \tau \}.
\end{equation}
Now we are in the place to define our quantum walk model. 
\begin{definition}[Simplicial quantum walks, version 2]\rm
\label{dfn-SQW2}
The quantum walk on $n$-dimensional simplicial complex $\mathcal{K}$ ($n\geq 2$) is defined as follows. 
\begin{enumerate}
\item Total Hilbert space: $\ell^2(\mathcal{K}^{n,n-1})$. Here the inner product is the standard inner product such that
	\[ \langle \psi,\phi \rangle=\sum_{ \bs{\sigma} \in \mathcal{K}^{n,n-1}} \overline{\psi(\bs{\sigma})}{\phi(\bs{\sigma})}. \]
Using $\pi_1$ and $\pi_2$, we decompose $\ell^2(\mathcal{K}^{n,n-1})$ by 
\begin{equation*}
\ell^2(\mathcal{K}^{n,n-1})=\bigoplus_{\tau\in\bra \mathcal{K}_{n-1} \ket} \mathcal{E}_{\tau}=\bigoplus_{\sigma\in\bra K_{n} \ket} \mathcal{F}_{\sigma},
\end{equation*}
where
	\[ \mathcal{E}_\tau=\{\psi \;|\; \bs{\sigma}\notin E_\tau \Rightarrow \psi(\bs{\sigma})=0\},\;
           \mathcal{F}_\sigma=\{\psi \;|\; \bs{\sigma}\notin F_\sigma \Rightarrow \psi(\bs{\sigma})=0\}\]
\item Unitary time evolution: Set $\hat{E}_\tau$ and $\hat{F}_\sigma$ as local unitary operators on $\mathcal{E}_\tau$ and $\mathcal{F}_\sigma$. 
Put $\hat{E}=\bigoplus_{\tau\in \bra \mathcal{K}_{n-1} \ket}\hat{E}_\tau$ and $\hat{F}=\bigoplus_{\sigma\in \bra K_{n} \ket}\hat{F}_\sigma$. 
Then the time evolution operator is defined by 
	\[ \hat{U}=\hat{F}\circ \hat{E}. \]
We call $\hat{E}$ and $\hat{F}$ {\em the first and second unitary operators}, respectively. 
\item Distribution: letting  a unit initial state be $\psi_0\in \ell^2(\mathcal{K}^{n,n-1})$, for each natural number $k$, we define the distribution 
$\mu_k^{(\psi_0)}:  \mathcal{K}_{n-1} \to [0,1]$ by 
	\[ \mu_k^{(\psi_0)}(|\tau|)=\sum_{\bs{\sigma}\in E_{\tau'}, |\tau'| = |\tau|}|(\hat U^k\psi_0)(\bs{\sigma})|^2. \]
\end{enumerate}
We shall call the collection $(\hat U, \ell^2(\mathcal{K}^{n,n-1}), \mu_\ast^{(\cdot)})$, or simply $\hat U$, a {\em simplicial quantum walk} on $\mathcal{K}$.
We write the phrase \lq\lq simplicial quantum walk" by {\em SQW} for short.
In particular, the phrase {\em SQW2} means the SQW in the sense of the current definition.
Time evolution of SQW2 is illustrated in Figure \ref{fig-SQW}.
\end{definition}

\begin{figure}[htbp]\em
\begin{minipage}{0.32\hsize}
\centering
\includegraphics[width=6.0cm]{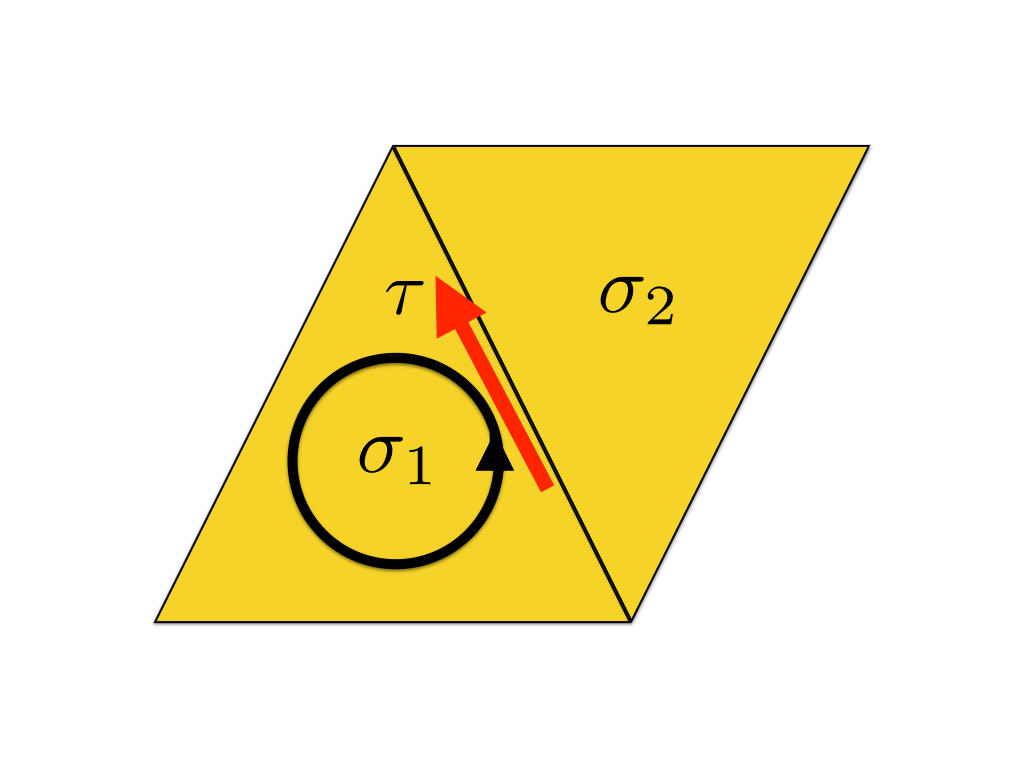}
(a)
\end{minipage}
\begin{minipage}{0.32\hsize}
\centering
\includegraphics[width=6.0cm]{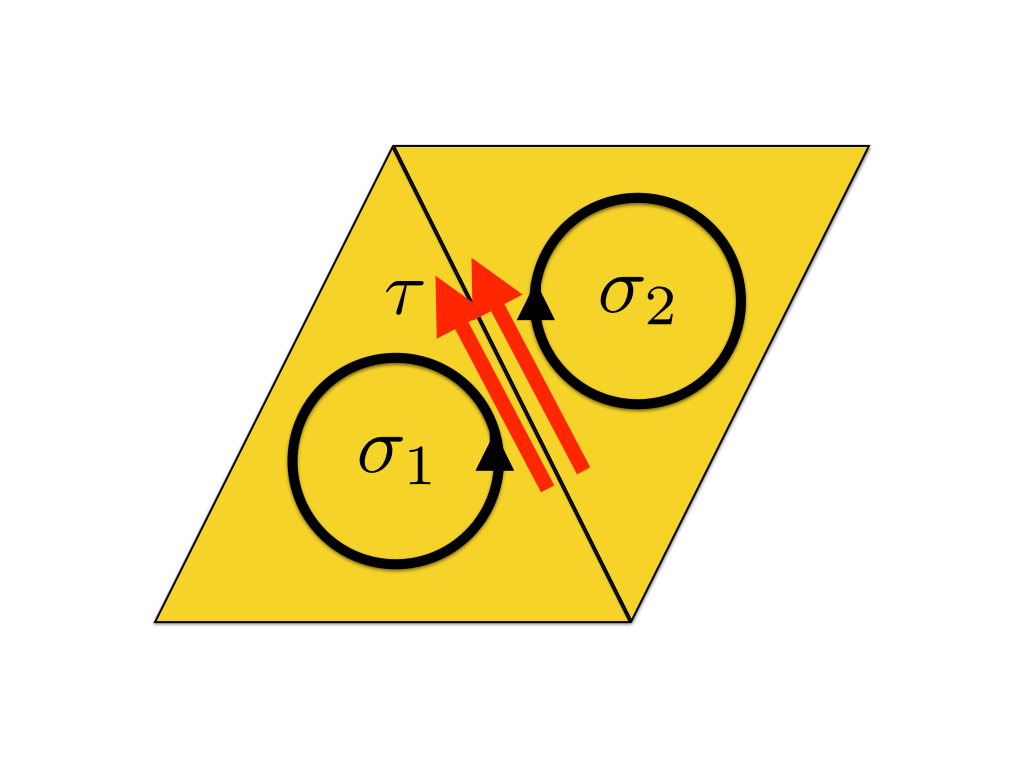}
(b)
\end{minipage}
\begin{minipage}{0.32\hsize}
\centering
\includegraphics[width=6.0cm]{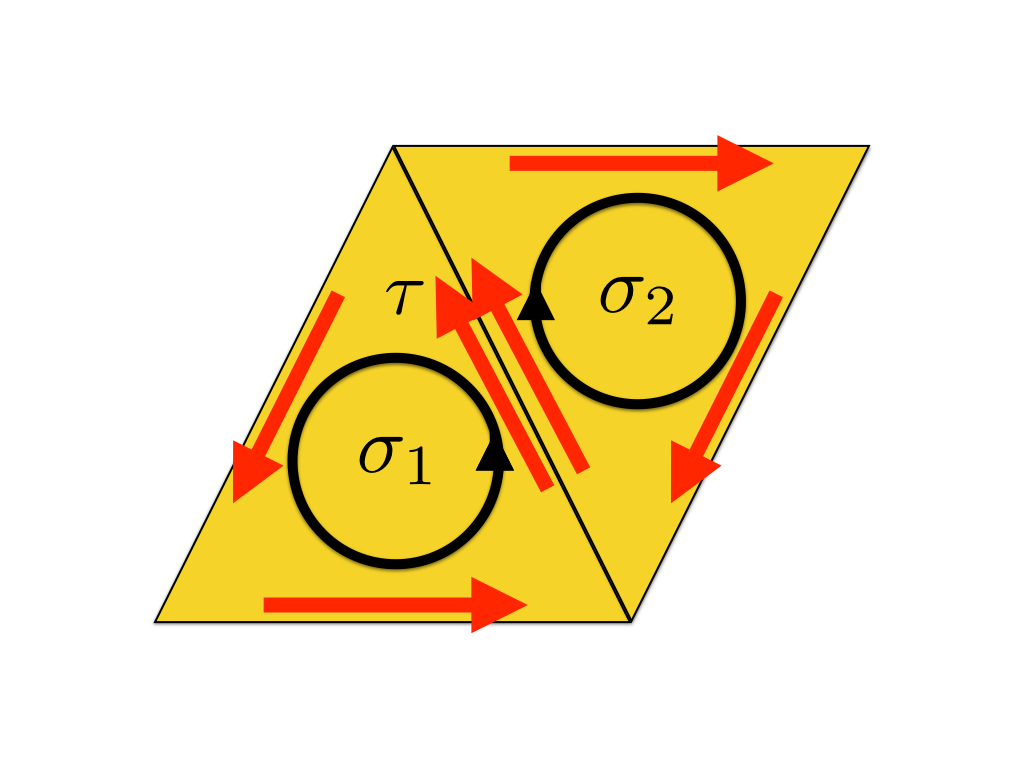}
(c)
\end{minipage}
\caption{Illustration of actions $\hat E$ and $\hat F$}
\label{fig-SQW}
Let $\mathcal{K}$ be the $2$-dimensional simplicial complex consisting of two triangles and their faces.
(a) : A state is assumed to be on the pair $(\sigma_1, \tau)\in \mathcal{K}^{2,1}$ of an oriented simplex $\sigma_1$ and a primary face $\tau$ with induced orientation.
(b) : Action of $\hat E$. 
The state transmits the adjacent simplex $\sigma_2$ whose orientation is uniquely determined so that the induced orientation on the primary face $|\sigma_1|\cap |\sigma_2|\equiv |\tau|$ is the same as $\tau$. 
(c) : Action of $\hat F$. 
States are diffused around simplices $\sigma_1$ and $\sigma_2$.
\end{figure}

In the current definition, the assumption $n\geq 2$ is crucial.
If no confusions arise, we assume that simplicial complexes have dimensions $n\geq 2$.
Several comments for the case $n=1$ are left to Remark \ref{rem-dim1}.
\begin{remark}\rm
\label{rem-comparison}
In the original version \cite{MOSver1}, SQW is defined by the procedure of Szegedy-type walks (e.g., \cite{HKSS2}), which is achieved by simple generalizations of coin and shift operators.
It makes sense even for $n=1$, namely, under several identifications of coin states, version 1 recovers Szegedy-type walks on graphs.
On the other hand, the construction indicates that the number of induced standard bases by each simplex (i.e., the dimension of {\it shift} space) is $(n+1)!$ for walks on $n$-dimensional simplicial complexes (note that the dimension of coin space for Szegedy walks on graphs is $2 = (1+1)!$).
Furthermore, the shift operator is a cyclic permutation on $n+1$ letters. 
These restrictions force us to spend quite high costs on mathematical and numerical studies of quantum walks.
\par
In the present version, the total space is constructed with attention to {\em orientations} on simplices and their faces, which reduce the dimension of coin spaces.
Comparing Figure \ref{fig-SQW} with Figure $1$ in \cite{MOSver1}, we can see that the evolution of $\hat U$ is similar to SQW in \cite{MOSver1}.
In the current case, the first and the second unitary operators $\hat E$ and $\hat F$ corresponding to coin and shift operators, respectively, can be chosen arbitrarily. 
In particular, these can be chosen to be involutions like Grover operators, in which sense SQW2 can be expected to be simplified from both mathematical and computational viewpoints.
Moreover, as shown below, SQW2 is unitary equivalent to quantum walks on graphs under appropriate choices of operators and simplicial complexes. 
%
These comparisons are listed in Table \ref{table-SQW}.
\end{remark}

\begin{table}[h]
\caption{Comparison with coined walks and SQWs on $\mathcal{K}$ with $\dim \mathcal{K} = n$}
\begin{center}
\begin{tabular}{|c|c|c|c|}
\hline
 & Coined walk & Version 1 \cite{MOSver1} & Version 2 \\
 & (e.g., \cite{HKSS2}) &  &  (Definition \ref{dfn-SQW2}) \\[2mm] \hline
total space & $\ell^2(A)$ & $\ell^2(\tilde K_n)$ & $\ell^2(\mathcal{K}^{n,n-1})$  \\
 & (on directed arcs) &  &   \\[2mm]
local coin & unitary on $\mathbb{C}^{\deg(u)}$ & a reflection operator &  unitary on $\mathcal{E}_\tau$  \\[2mm]
local shift & flip-flop & cyclic permutation  & unitary on $\mathcal{F}_\sigma$ \\
 	& & on $\{0,1,\cdots n\}$ & \\[2mm]
$\dim$(shift space) & $2$ & $(n+1)!$ & $2(n+1)$ \\[2mm]
$n=1$ & $-$ & coined walk & coined walk on \\
 & & under several identifications & double graph \\[2mm]
\hline
\end{tabular}
\end{center}
\label{table-SQW}
\end{table}%

As seen in the next subsections, the new version of simplicial quantum walks can be seen as quantum walks on corresponding graphs, which lets studies of dynamics of walkers much simpler.

\subsection{Equivalence to bipartite walks}

Here we show that our present simplicial quantum walks are interpreted as quantum walks on bipartite graphs, called bipartite walks. 
The bipartite walk is defined as follows. 
\begin{definition}[\cite{Sze}]\rm
Let $(X\sqcup Y,E)$ be a connected bipartite graph and $\ell^2(E)$ be the Hilbert space induced by $E$, where $\sqcup$ denotes the disjoint union of two sets. 
We decompose $\ell^2(E)$ into 
\begin{equation*}
\ell^2(E)=\bigoplus_{x\in X} {\mathcal{E}}'_x=\bigoplus_{y\in Y} {\mathcal{F}}'_y,
\end{equation*}
where 
\begin{equation*}
{\mathcal{E}}'_x=\{\psi\in \ell^2(E) \;|\; X(e)\neq x \Rightarrow \psi(e)=0\},\quad {\mathcal{F}}'_y=\{\psi\in \ell^2(E) \;|\; Y(e)\neq y \Rightarrow \psi(e)=0\}
\end{equation*}
and $X(e)$ and $Y(e)$ are the end vertices of $e$ in $X$ and $Y$, respectively. 
Setting local unitary operators ${\hat{E}}'_x$ and ${\hat{F}}'_y$ on ${\mathcal{E}}'_x$ and ${\mathcal{F}}'_y$, respectively,
we define the following unitary operator on $\ell^2(E)$:
\begin{equation*}
\hat{B}'= \left(\bigoplus_{y\in Y} {\hat{F}}'_y \right) \circ \left(\bigoplus_{x\in X} {\hat{E}}'_x \right).
\end{equation*}
We call the walk driven by $\hat B'$ a {\em bipartite walk} on $(X\sqcup Y, E)$. 
\end{definition}
Now we introduce a special bipartite graph induced by $\mathcal{K}^{n,n-1}$ as follows. 
\begin{definition}\rm
We define a bipartite graph $G_\cap = G_\cap (\mathcal{K}^{n,n-1})=(X_E\sqcup X_F, E(G_\cap))$ induced by $\mathcal{K}^{n,n-1}$ as follows. 
See (\ref{set-E-F}) for the definition of $E_\tau$ and $F_\sigma$.
\begin{align}
\notag
	X_E = \{E_\tau \;|\; \tau\in \bra \mathcal{K}_{n-1} \ket\},&\;X_F = \{F_\sigma \;|\; \sigma\in \bra \mathcal{K}_{n} \ket\} \\
\label{edge-def}
        E_\tau F_\sigma\in E(G_\cap) \quad \stackrel{def}{\Leftrightarrow} \quad & E_\tau\cap F_\sigma \neq \emptyset \;\mathrm{in}\; \mathcal{K}^{n,n-1}. 
\end{align}
We call this graph an {\em induced bipartite graph of $\mathcal{K}^{n,n-1}$}. 
An example of the graph $G_\cap (\mathcal{K}^{n,n-1})$ is illustrated in Figure \ref{fig-bipartite}.
\end{definition}

\begin{figure}[htbp]\em
\begin{minipage}{0.5\hsize}
\centering
\includegraphics[width=8.0cm]{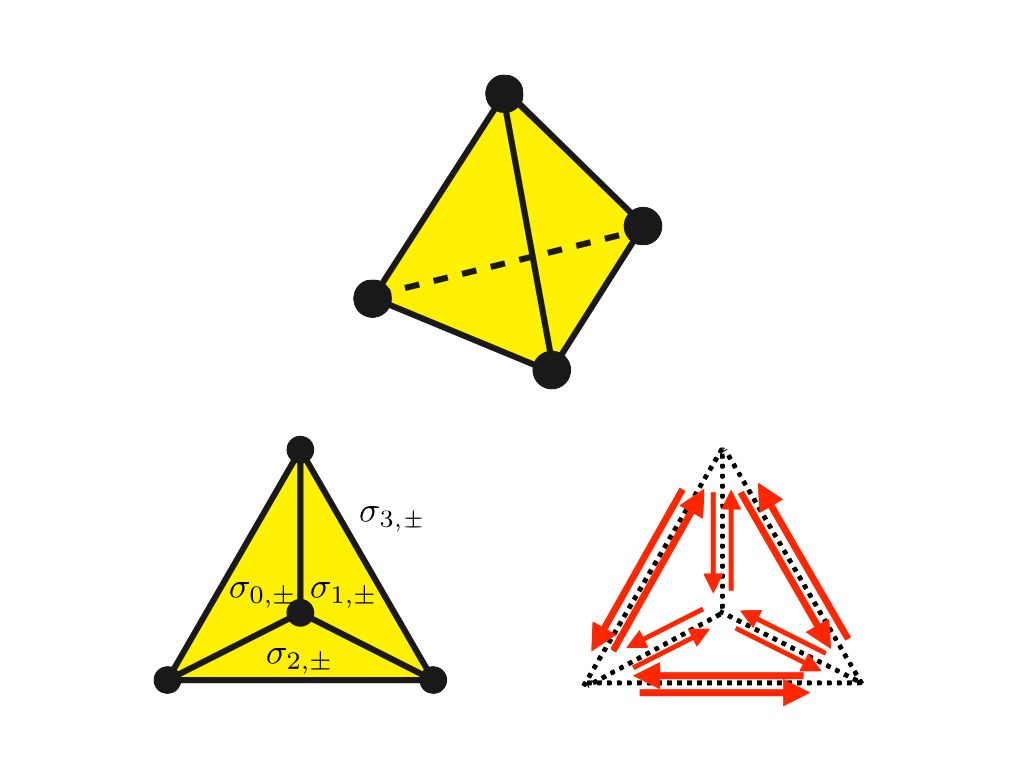}
(a)
\end{minipage}
\begin{minipage}{0.5\hsize}
\centering
\includegraphics[width=8.0cm]{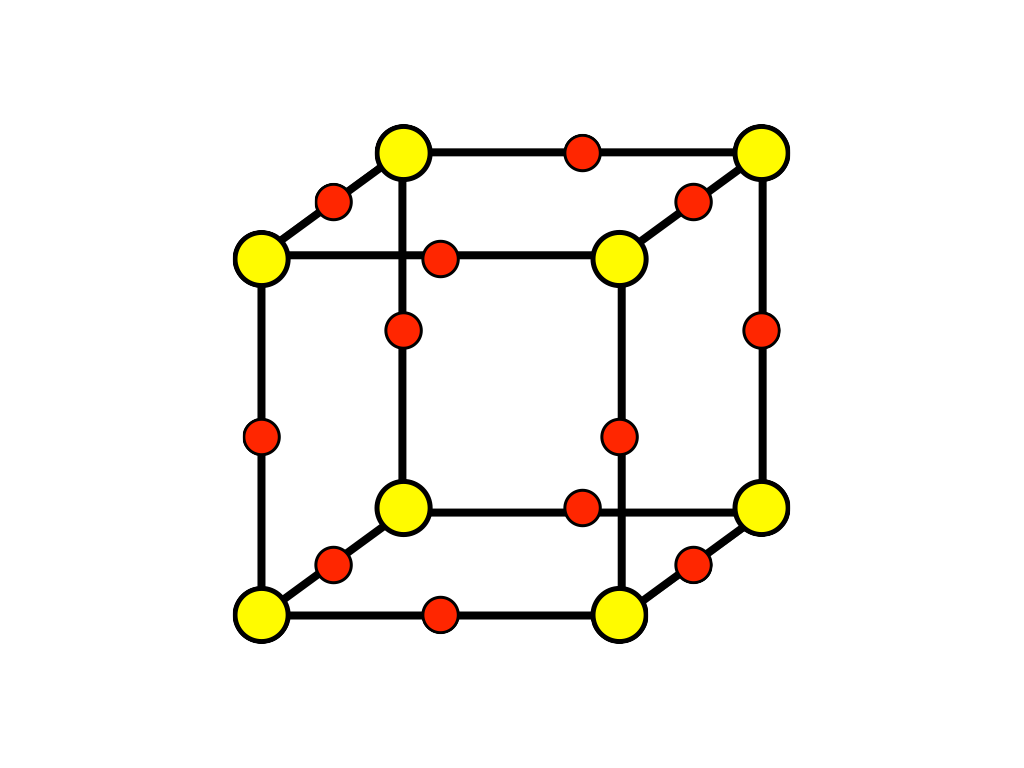}
(b)
\end{minipage}
\caption{Simplicial complex $\mathcal{K}$ and induced bipartite graphs $G_\cap (\mathcal{K}^{n,n-1})$}
\label{fig-bipartite}
(a) : A simplicial complex $\mathcal{K}$ describing a tetrahedron with a cavity. 
Eight oriented $2$-simplices and twelve oriented $1$-simplices are stored.
(b) : The induced bipartite graph $G_\cap (\mathcal{K}^{2,1})$.
Each vertex corresponds to either a $2$-simplex (yellow stored in $X_F$) or a $1$-simplex (red stored in $X_E$).
The connection is determined by the rule (\ref{edge-def}).
\end{figure}

Note that the connectivity of this graph is not ensured in general.
We can take a bijective map $\eta:\mathcal{K}^{n,n-1}\to E(G_\cap)$ such that
	\[ \eta((\sigma,\tau))=F_\sigma E_\tau.  \]
The inverse map is 
	\[ \eta^{-1}(F_{\sigma}E_{\tau})=(\sigma,\tau). \]
Let $\ell^2(E(G_\cap))$ be the $\ell^2$-summable Hilbert space. 
Using the unitary map $\mathcal{U}_\eta: \ell^2(\mathcal{K}^{n,n-1})\to \ell^2(E(G_\cap))$ defined by
\begin{equation*}
(\mathcal{U}_\eta\psi)(E_{\tau}F_{\sigma})= \psi(\sigma,\tau),
\end{equation*}
we can interpret this walk as a walk $\hat{U}$ on $(X_E\sqcup Y_F, E(G_\cap))$: 
\begin{equation*}
\hat{U}=\mathcal{U}_{\eta}^{-1} \hat{B}' \mathcal{U}_\eta.
\end{equation*}
To see more detail of $\hat{U}$, let us consider   
	\begin{align*} 
        \mathcal{U}_\eta \hat{U} \mathcal{U}_\eta^{-1}
        	&=  (\mathcal{U}_\eta \hat{F} \mathcal{U}_\eta^{-1})\cdot (\mathcal{U}_\eta \hat{E} \mathcal{U}_\eta^{-1}).
        \end{align*}
The first term of right-hand side (RHS) is the direct sum of the local unitary operators $\{\mathcal{U}_\eta\hat{F}_\sigma\mathcal{U}_{\eta}^{-1}\}$'s based on the decomposition $E(G_\cap) = \sqcup_{y\in Y_F}\{e\in E(G_\cap) \;|\; Y(e)=y\}$, while the second term in RHS 
is the direct-sum of the local unitary operators $\{\mathcal{U}_\eta\hat{E}_\tau\mathcal{U}_{\eta}^{-1}\}$'s  
based on the decomposition $E(G_\cap) = \sqcup_{x\in X_E}\{e\in E(G_\cap) \;|\; X(e)=x\}$. 
Therefore, this is nothing but a bipartite walk on $G_\cap(\mathcal{K}^{n,n-1})$. 
\par
Here summarize the above arguments. 
\begin{proposition}
\label{graphinterpretation}
The quantum walk on simplicial complex $\mathcal{K}$ is unitary equivalent to the bipartite walk on the induced bipartite graph $G_\cap(\mathcal{K}^{n,n-1})$. 
\end{proposition}

\subsection{Simplicial quantum walks on orientable simplicial complexes}
SQW2 defined in Definition \ref{dfn-SQW2} are turned out to be equivalent to quantum walks on bipartite graphs.
Here we consider SQW2 on simplicial complexes $\mathcal{K}$ with several geometric constraints on $\mathcal{K}$.
In particular, we pay attention to {\em orientable} simplicial complexes (see Appendix \ref{section-orientability} for orientability of simplicial complexes), in which case SQW2 on them can be expressed as simpler quantum walks on graphs with special geometric properties.

\subsubsection{Coined walks and bipartite walks on corresponding graphs}
First consider {\em coined (quantum) walks} on connected graphs. 
The coined walk, which has been traditionally studied, is defined by the pair of connected graph $G=(V,E)$ and sequence of local unitary operators assigned to each vertex $\{\hat{C}_u\}_{u\in V}$. 
Let $A$ be the symmetric (directed) arc set induced by $E$. The terminus and the origin of $a\in A$ are denoted by $t(a)$ and $o(a)\in V$, respectively. 
The inverse arc $a$ is denoted by $\bar{a}$ such that $o(a)=t(\bar{a})$ and $t(a)=o(\bar{a})$. 
The degree of $u\in V$ is $\deg (u)=\sharp \{ a\in A \;|\; t(a)=u \}$. 
Due to the symmetricity of the arc, it follows that $\deg (u)=\sharp \{ a\in A \;|\; o(a)=u \}$. 
\begin{definition}\rm
Let $G=(V,E)$ be a connected graph and $\ell^2(A)$ be the Hilbert space induced by the symmetric arcs of $G$. 
We decompose $\ell^2(A)$ into 
\begin{equation*}
\ell^2(A)=\bigoplus_{u\in V} \mathcal{C}_u,
\end{equation*}
where $\mathcal{C}_u=\{ \psi\in \ell^2(A) \;|\; t(a)\neq u \Rightarrow \psi(a)=0 \}$. 
Setting a local (arbitrary) unitary operator $\hat{C}_u$ on $\mathcal{C}_u$, we define the following unitary operator
\begin{equation*}
\hat{\Gamma}=\hat{S}\hat{C},
\end{equation*}
where $\hat{C}=\bigoplus_{u\in V} \hat{C}_u$ and $(\hat{S}\psi)(a)=\psi(\bar{a})$. 
We call the walk driven by $\hat{\Gamma}$ {\em a coined (quantum) walk}. 
\end{definition}
The subdivision graph of $G$ is denote by $S(G)$ so that 
\begin{align*} 
V(S(G)) &= V\sqcup E, \\ 
E(S(G)) &= \{ue\;|\; u\in V,\; e\in E,\; u \text{ is an endpoint of } e \text{ in }G\}. 
\end{align*}
Note that the degree of $e\in E$ is identically two in $S(G)$. Thus on the bipartite walk on $S(G)=(V\sqcup E,E(S(G)))$, 
the local unitary operator $\hat{F}_e$ is two-dimensional unitary operator for $e\in E$. 
The following proposition presents that every coined walks on graphs can be described by bipartite walks on their subdivision graphs.
In other words, the underlying graph where bipartite walks admit coined walks must be a subdivision graph.
This statement plays an essential role to connect coined walks to SQWs on several simplicial complexes.

\begin{proposition}[\cite{PS}]
\label{bipartitecoind}
Let $G=(V,E)$ be a connected graph and $\hat{B}=\bigoplus_{e\in E}\hat{F}_e\cdot \bigoplus_{u\in V}\hat{E}_u$ be the time evolution operator of a bipartite walk on the subdivision graph $S(G)$. 
Assume 
	\[ F_e\cong \begin{bmatrix} 0 & 1 \\ 1 & 0 \end{bmatrix}\] 
for every $e\in E$. 
Then there exists a coined walk $\hat{\Gamma}$ on $G$ so that it is unitary equivalent to $\hat{B}$, that is, 
 	\[ \hat{\Gamma}=\mathcal{V}^{-1} \hat{B} \mathcal{V}, \] 
where $\mathcal{V}$ is a unitary map from $\ell^2(A(G)) \to \ell^2(E(S(G)))$ with 
	\[ (\mathcal{V}\psi)(ue)=\psi(a) \mathrm{\;with\;} t(a)=u,\; |a|=e. \]
Here $|a|\in E$ is the support of $a\in A$. 
\end{proposition}

\subsubsection{Equivalence of simplicial quantum walks with geometric constraints}
In this section, we show that SQW2 on every orientable simplicial complexes can be reduced to coined walks on associated graphs.
This fact will make the spectral analysis simple. 
In this subsection, let $\mathcal{K}$ be an orientable $n$-dimensional simplicial complex.
We also assume that all the local unitary coins of the first operators are the Grover operators. 
A matrix expression of the $m$-dimensional Grover operator is denoted by 
\begin{equation*}
\frac{2}{m}J_m-I_m,
\end{equation*}
where $J_m$ is the $m$-dimensional all $1$ matrix and $I_m$ is the identity matrix. 
Remark that the first unitary operator of SQW2 on $\mathcal{K}$ in this setting becomes  
\begin{equation*}
\hat{E}= \left( \bigoplus_{\tau \in \bra \mathcal{K}_{n-1}\ket,\ |\tau|\not \in B_{n-1}(\mathcal{K})} \begin{bmatrix} 0 & 1 \\ 1 & 0 \end{bmatrix}  \right) 
\oplus 
\left( \bigoplus_{\tau \in \bra \mathcal{K}_{n-1}\ket,\ |\tau| \in B_{n-1}(\mathcal{K})} I_1 \right),
\end{equation*}
where $B_{n-1}(\mathcal{K})$ is the set of boundary simplices in $\mathcal{K}$ consisting of $(n-1)$-simplices (cf. Appendix \ref{appendix-cpx}). 
\begin{definition}\rm
(Associated graph of $\mathcal{K}$) 
Let $\mathcal{K}$ be an $n$-dimensional orientable simplicial complex. 
The {\em associated graph of $\mathcal{K}$} is denoted by $G_a(\mathcal{K})=(V,E)$ with 
\begin{align*}
        V &= \mathcal{K}_n,  \\ 
        \sigma\sigma' \in E &\quad \stackrel{def}{\Leftrightarrow}\quad \sigma \cap \sigma'\in \mathcal{K}_{n-1}\setminus \{\emptyset\}.
\end{align*}
\end{definition}

\begin{definition}\rm
The {\em duplication graph} of a simple graph $G=(V,E)$ is a bipartite graph $D(G)=(V\sqcup \phi(V),E_2)$, where 
$\phi(V)$ is a copy of $V$ and 
\begin{align*}
u\phi(v)\in E_2\quad \stackrel{def}{\Leftrightarrow}\quad uv\in E.
\end{align*}
\end{definition}
\begin{lemma}
\label{lem-ori}
Assume that $\mathcal{K}$ be an $n$-dimensional simplicial complex which is strongly connected and orientable. 
Then we have
\begin{equation*}
S\circ D\circ G_{a}(\mathcal{K}) \bullet B_{n-1} \cong G_{\cap}(\mathcal{K}^{n,n-1}),
\end{equation*}
where \lq\lq $\bullet B_{n-1}$" is an operation to $S\circ D\circ G_{a}(\mathcal{K})$ defined as follows: 
if $|\sigma|\in \mathcal{K}_n \subset V(S\circ D\circ G_{a}(\mathcal{K}))$ has $r$ primary faces which belong to $B_{n-1}$, 
then we replace $|\sigma|$ and $\phi(|\sigma|)$ into $r$-bunches, respectively, that is, 
we add $r$ one-paths to $|\sigma|$ and its copy $\phi(|\sigma|)$, respectively. 
The resulting graph is illustrated in Figure \ref{fig-SDGB}.
\end{lemma}
\begin{figure}[htbp]\em
\begin{minipage}{0.32\hsize}
\centering
\includegraphics[width=6.0cm]{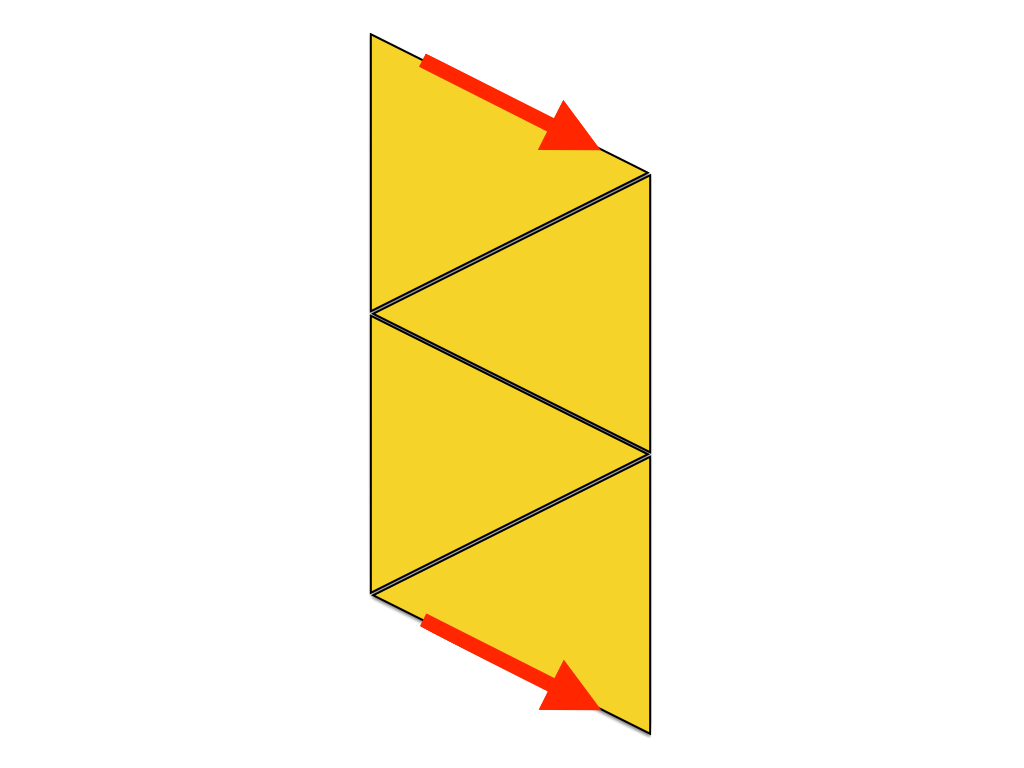}
(a)
\end{minipage}
\begin{minipage}{0.32\hsize}
\centering
\includegraphics[width=6.0cm]{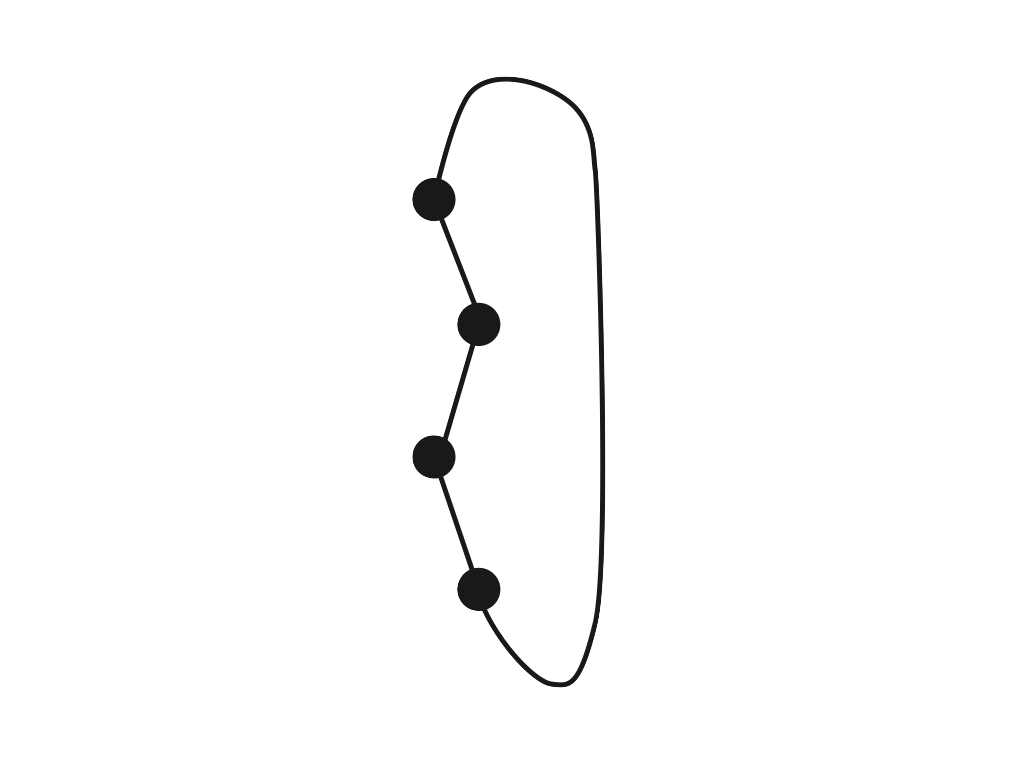}
(b)
\end{minipage}
\begin{minipage}{0.32\hsize}
\centering
\includegraphics[width=6.0cm]{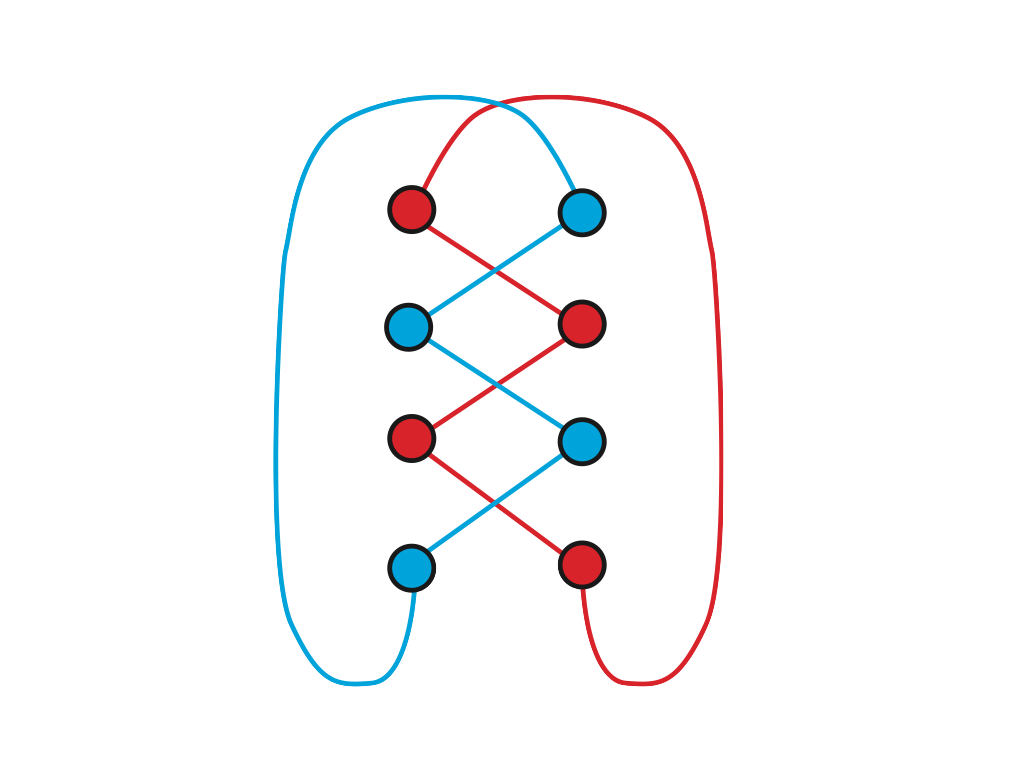}
(c)
\end{minipage}
\begin{minipage}{0.32\hsize}
\centering
\includegraphics[width=6.0cm]{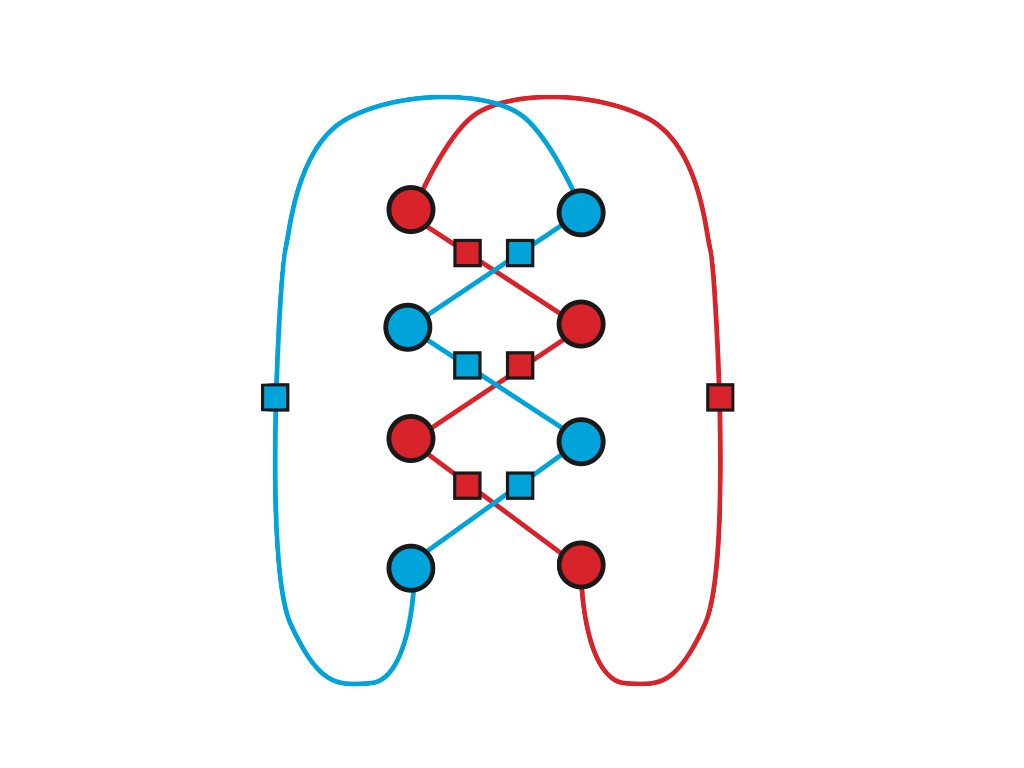}
(d)
\end{minipage}
\begin{minipage}{0.32\hsize}
\centering
\includegraphics[width=6.0cm]{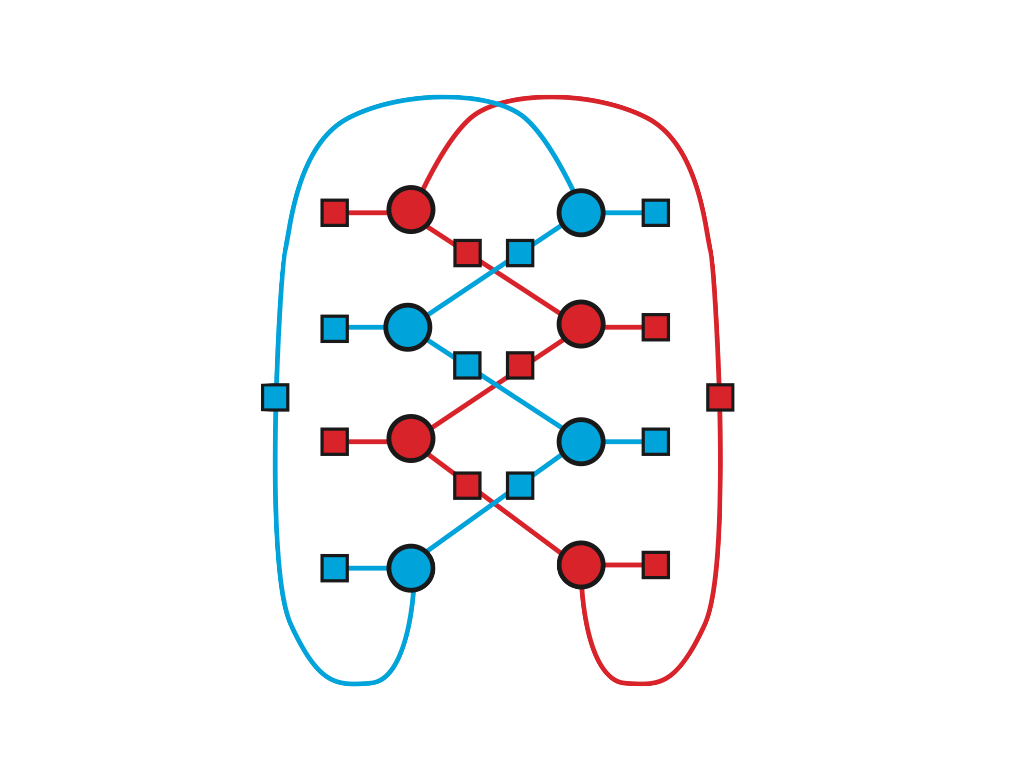}
(e)
\end{minipage}
\caption{Orientable simplicial complex $\mathcal{K}$ and the graph $S\circ D\circ G_{a}(\mathcal{K}) \bullet B_{n-1}$}
\label{fig-SDGB}
(a) : Original (orientable) simplicial complex $\mathcal{K}$, where the orientation is ignored. 
The arrow means the identification of edge.
(b) : The associated graph $G_a(\mathcal{K})$.
(c) : The duplication graph $D\circ G_{a}(\mathcal{K})$ of $G_{a}(\mathcal{K})$.
(d) : The subdivision graph $S\circ D\circ G_{a}(\mathcal{K})$ of $D\circ G_{a}(\mathcal{K})$.
(e) : The final graph $S\circ D\circ G_{a}(\mathcal{K}) \bullet B_{n-1}$.
Note that this graph is bipartite with vertex sets $V_1 = \text{(circles)}$ and $V_2=\text{(squares)}$.
\end{figure}
\begin{proof}
Let $(V_1,E_1) = G_{\cap}(\mathcal{K}^{n,n-1})$ and $(V_2,E_2) = S\circ D\circ G_{a}(\mathcal{K}) \bullet B_{n-1}$. 
The vertex set of $D\circ G_a(\mathcal{K})$ is constructed by $\mathcal{K}_n \sqcup \phi(\mathcal{K}_n)$. 
The edge set of $D\circ G_a(\mathcal{K})$ is expressed by
	\[  Int_{n-1}:=\{|\sigma|\phi(|\sigma|') \; |\; |\sigma|\cap |\sigma|' \in \mathcal{K}_{n-1}\setminus \{\emptyset\} \}. \]
Therefore the vertex set of $S\circ D\circ G_a(\mathcal{K})$ is described by 
	\[ V_1=\mathcal{K}_n \sqcup \phi(\mathcal{K}_n) \sqcup Int_{n-1} \sqcup B_{n-1} \sqcup \phi(B_{n-1}). \]
\par
The vertex set $V_2$ is decomposed into 
	\[ V_2=\{ E_\tau \;|\; \tau\in \bra \mathcal{K}_{n-1} \ket \} \sqcup \{ F_\sigma \;|\; \sigma\in \bra \mathcal{K}_n \ket \} \]
and $E_\tau$ and $F_\sigma$ is connected in $G_\cap$ if and only if $E_\tau \cap F_\sigma \neq \emptyset$, which is 
equivalent to $\sigma\triangleright \tau$. 
We present $V_2$ by
\begin{equation*}
V_2\cong \bra \mathcal{K}_n \ket \sqcup \bra \mathcal{K}_{n-1} \ket.
\end{equation*}
From now on we regard $V_2$ as the above RHS. 
\par
We define the following map $\xi:V_1\to V_2$ so that 
\begin{align*}
        \xi(|\sigma|) &= (|\sigma|,+),\; \xi(\phi(|\sigma|)) = (|\sigma|,-), \; \mathrm{for} \; |\sigma|\in \mathcal{K}_n\\
        \xi(|\sigma|\phi(|\sigma|')) &= \tau \;\mathrm{with}\; (|\sigma|,+),\; (|\sigma|',-)\triangleright \tau, \; \mathrm{for} \; |\sigma|\phi(|\sigma|')\in Int_{n-1} \\
        \xi(b_j) &= \tau_j,\; \xi(\phi(b_j)) = \tau_j'.
\end{align*}
Here $\{b_1,\dots,b_r\}$ is the set of bunches on $|\sigma|\in \mathcal{K}_n$ and 
$\{\tau_1,\dots,\tau_r\}$ is the set of boundaries with $(|\sigma|,+)\triangleright \tau_j$, and 
$\{\tau_1',\dots,\tau_r'\}$ is the set of boundaries with $(|\sigma|,-)\triangleright \tau_j'$, where $(|\sigma|, \pm)$ denotes the simplex $|\sigma|$ with one choice of orientations. 
We can check that this map $\xi$ is bijective due to the orientability. 
From this bijection map, it is easy to see that $uv\in E_1$ if and only if $\xi(u)\xi(v)\in E_2$. 
This completes the proof. 
\end{proof}
\begin{remark}
The proof indicates that, if $\mathcal{K}$ is not orientable, then the associated graph $G_\cap (\mathcal{K}^{n,n-1})$ cannot be expressed as a duplication graph (with bunches corresponding to boundaries).
In other words, geometric obstructions such as non-orientability give responses to complexity of associated graphs.
\end{remark}

If the orientable simplicial complex $\mathcal{K}$ has {\em no boundaries}, SQW2 on $\mathcal{K}$ can be expressed much simpler as follows.

\begin{theorem}
\label{thm-equiv-duplicate}
Let $\mathcal{K}$ be an $n$-dimensional orientable simplicial complex without boundaries. 
The time evolution of SQW2 on $\mathcal{K}$ driven by the Grover operator is denoted by $U$,
while that of the coined walk on $D\circ G_a(\mathcal{K})$ driven by the Grover operator is denoted by $\Gamma$. 
Then we have 
\begin{equation*}
U=\mathcal{W}^{-1}\Gamma^T \mathcal{W},
\end{equation*}
where $\ast^T$ denotes the transpose of vectors or operators. 
Here the unitary map $\mathcal{W}:\ell^2(\mathcal{K}^{n,n-1})\to \ell^2(A(D\circ G_a(\mathcal{K})))$ is given by the following: 
for $t(a)=|\sigma|$, $o(a)=\phi(|\sigma|')$ with $(|\sigma|,+),(|\sigma|',-) \triangleright \tau$, 
\begin{equation*} 
(\mathcal{W}\psi)(a) = \psi( (|\sigma|,+), \tau), \quad (\mathcal{W}\psi)(\bar{a}) = \psi( (|\sigma|,-), \tau).
\end{equation*}
\end{theorem}
The inverse map $\mathcal{W}^{-1}$ is 
\begin{equation*}
(\mathcal{W}^{-1}\varphi)((|\sigma|,\pm),\tau) 
= \begin{cases} 
	 \varphi(a) & \text{ for $(|\sigma|, +)$,} \\ 
	 \varphi(\bar{a}) & \text{ for $(|\sigma|, -)$,} 
\end{cases}
\end{equation*}
where $t(a)=|\sigma|$, $o(a)=|\sigma|'$ with $|\sigma|\cap |\sigma|'=|\tau|$ in $\mathcal{K}$. 
\begin{proof}
From Proposition~\ref{graphinterpretation}, SQW2 is described by a bipartite walk on the induced bipartite graph $G_\cap$. 
By Lemma.~\ref{lem-ori}, since $\mathcal{K}$ is orientable and dose not have boundaries, the induced graph is described by $S\circ D\circ G_a(\mathcal{K})$. 
This walk is therefore expressed by a bipartite walk on $S\circ D \circ G_a(\mathcal{K})$. 
Since we adopt the Grover operators to local unitary operators, then Proposition~\ref{bipartitecoind} implies that this walk is isomorphic to 
coined walk on $D \circ G_a(\mathcal{K})$. 
This completes the proof. 
\end{proof}

\begin{remark}\rm
\label{rem-dim1}
If we define the operator $U$ in Theorem \ref{thm-equiv-duplicate} on $1$-dimensional simplicial complexes; namely on graphs, the resulting quantum walk is not the Grover walk on graphs, but actually the walk on the {\em double graphs} associated with given graphs.
The \lq\lq double" structure stems from the orientability of vertices.
By the definition, vertices of graphs are trivially orientable and hence, according to the definition of SQWs, each vertex admits two independent states.
However, in the case of original coined walks, vertices are not distinguished by orientations, which generates the gap of states associated with vertices between coined walks and SQW2 on graphs.
\end{remark}

Summarizing arguments in this section, we have the following quantum walks and unitary equivalence between them.
\begin{itemize}
\item Simplicial quantum walk $\hat U$ on a simplicial complex $\mathcal{K}$ (SQW2, Definition \ref{dfn-SQW2}) and bipartite walk $\hat B$ on the induced bipartite graph $G_\cap(\mathcal{K}^{n,n-1})$ (Proposition \ref{graphinterpretation}).
\item Coined quantum walk $\hat \Gamma$ on a graph $G$ and bipartite walk $\hat B$ on the subdivision graph $S(G)$ (Proposition \ref{bipartitecoind}).
\item Grover-driven SQW2 $U$ on an orientable simplicial complex with no boundary $\mathcal{K}$ and Grover-driven coined quantum walk $\Gamma$ on the associated duplication graph $D\circ G_a(\mathcal{K})$ (Theorem \ref{thm-equiv-duplicate}).
\end{itemize}
The last unitary equivalent quantum walks are main actors of quantum search problems in the next section.

%
%
\section{Simplicial quantum search}
\label{section-search}
In this section, we consider quantum search problem on simplicial complexes, which we shall call {\em simplicial quantum search} problem.
As an analogy of quantum search on graphs, our main problem is given as follows.
\begin{problem}[Simplicial quantum search problem]
\label{prob-search}
For a given $n$-dimensional strongly connected simplicial complex $\mathcal{K}$ with $\sharp \mathcal{K}_{n-1} = N$, fix an $(n-1)$-dimensional face as {\em a marked simplex}.
Then find the marked simplex with high probability independent of $N$ within a time $O(\sqrt{N})$, or at the latest $o(N)$.
\end{problem}

This is an analogue of quantum search problems on graphs.
Note that, in the case of quantum search on complete graphs or Grover's Algorithm (e.g., \cite{P2013}), the order of time complexity $O(\sqrt{N})$ for search problems in $N$ database entries (like vertices) is optimal.
In our case, the entries correspond to primary faces in the $n$-dimensional simplicial complex $\mathcal{K}$, and hence the order $O(\sqrt{N})$ will be the standard (and possibly optimal) benchmark for achieving quantum speed-up in search problems.
\par
Theorem \ref{thm-equiv-duplicate} claims that simplicial quantum walks on orientable simplicial complexes with no boundary driven by Grover operators are equivalent to Grover-type coined walks on graphs with duplication structures.
This observation gives us possibilities of quantum search considerations on simplicial complexes from the viewpoint of search on graphs.

\subsection{Setting and the main result}
Let $K_n$ be the complete graph with $n$ vertices. The clique complex of $K_n$ is denoted by $X(K_n)$. The $m$-skeleton of $X(K_n)$ is denoted by $X(K_n)^{(m)}$ $(m=0,\dots,n-1)$.  (See Appendix \ref{appendix-cpx} for notions of skeletons and clique complexes).
Our underlying simplicial complex $\mathcal{K}$ here is as follows:
\begin{equation}
\label{main-cpx}
\mathcal{K} := X(K_{n+2})^{(n)},
\end{equation}
namely, the $n$-skeleton of the clique complex $X(K_{n+2})$ of the complete graph $K_{n+2}$ with $(n+2)$ vertices.
This complex is regarded as a triangulation of $n$-dimensional unit sphere 
\begin{equation*}
S^n = \left\{(x_0, \cdots, x_n)\in \mathbb{R}^{n+1} \mid \sum_{i=0}^n x_i^2 = 1\right\}.
\end{equation*}
Indeed, $\mathcal{K}$ is the simplicial complex generated by an $(n+1)$-simplex removing the $(n+1)$-simplex itself.
In particular, it is strongly connected and orientable. 
Moreover, $B_{n-1}(\mathcal{K}) = \emptyset$.
We will therefore apply Theorem~\ref{thm-equiv-duplicate} to (\ref{main-cpx}) when we consider SQW2 on this $n$-dimensional simplicial complex $\mathcal{K}$ driven by the Grover operator\footnote{
Our search problem is topologically regarded as {\em the quantum search on the unit sphere}.
}.
From now on, we propose a quantum search on this simplicial complex with a marked simplex $\tau_*\in \mathcal{K}_{n-1}$ as follows: 
First, we set the indicator function of $\tau_*$, $f:\langle \mathcal{K}_{n-1} \rangle\to \{0,1\}$ such that
	\[ f(\tau)=\begin{cases} 1 & \text{: $|\tau|=\tau_*$,} \\ 0 & \text{: $|\tau|\neq \tau_*$.} \end{cases} \]
Secondly setting the total Hilbert space by $\ell^2(\mathcal{K}^{n,n-1})$ with $\dim \ell^2(\mathcal{K}^{n,n-1})=2(n+1)(n+2)\equiv 4N$, 
we adopt the following time evolution operator of the SQW2 which drives the quantum search. 

\begin{equation}
\label{qs}
\hat U_* = \left( \bigoplus_{ \sigma \in \bra\mathcal{K}_n \ket}\hat{F}_\sigma \right) \circ \left( \bigoplus_{\tau \in \bra \mathcal{K}_{n-1} \ket}\hat{E}_\tau \right)
\end{equation} 
\par
The first and the second terms correspond to the shift and coin operators of coined walk, respectively.
In coined walks on graphs with a marked vertex set $M\subset V$, the local coin operator depends on vertex so that
\begin{equation*}
\hat{C}_u = \frac{1+(-1)^{\bs{1}_{M}(u)}}{\deg(u)}J_{\deg(u)} - I_{\deg(u)},
\end{equation*}
where
\begin{equation*}
{\bs{1}_{M}(u)} = \begin{cases}
	1 & \text{ if $u\in M$,}\\
	0 & \text{ if $u\not \in M$,}
\end{cases}
\end{equation*}
which means the coin operator is perturbed at the marked vertices. Here $J_k$ is the all one $k$ dimensional matrix. 
Moreover it is easy to recognize that the shift operator is expressed by the direct-sum of $2$-dimensional Grover operators. 
\par
We extend this idea to our search problem: the \lq\lq degree" corresponds to \lq\lq $2$", since the simplicial complex is assumed to be oriented with no boundaries.
Therefore, the local unitary operator of $\hat{E}_\tau$ is expressed by
	\[ \hat{E}_\tau=\frac{1+(-1)^{f(\tau)}}{2} J_2-I_2, \] 
that is, 
\begin{equation*}
\label{perturb} 
	\hat{E}_{\tau} = 
  	\begin{cases} 
        \begin{bmatrix} 0 & 1 \\ 1 & 0 \end{bmatrix} \text{: $|\tau|\neq \tau_\ast$}, \\
        \\
        -\begin{bmatrix} 1 & 0 \\ 0 & 1 \end{bmatrix} \text{: $|\tau|= \tau_\ast$}.
        \end{cases}  
\end{equation*}
We thus set the local unitary operators $\{\hat{F}_\sigma\}_{ \sigma \in \bra\mathcal{K}_n \ket}$ corresponding to the shift operator as the $(n+1)$-dimensional Grover operator. 

We start the walk from the following uniform initial state: 
\begin{equation}
\label{initial-state}
\psi_{IN}(\sigma,\tau) :=1/2\sqrt{N},\quad \forall (\sigma, \tau)\in \langle \mathcal{K}_n\rangle \times \langle \mathcal{K}_{n-1}\rangle.
\end{equation}
Our aim is to estimate the time $t_f$ which produces a high probability that we observe the subspace of $\ell^2(\mathcal{K}^{n,n-1})$ 
spanned by $\{(\sigma,\tau)\in \mathcal{K}^{n,n-1} \;|\; |\tau|=\tau_\ast \}$, 
that is, the time $t_f$ is expressed by  
	\[ \sum_{|\tau|=\tau_\ast}|(U^{t_f}\psi_{IN})(\sigma,\tau)|^2 = p_f, \]
where $p_f$ is a high probability, for example, $p_f>1/2$. 
Indeed, we have the following theorem. 
\begin{theorem}
\label{thm-search}
Let $\mathcal{K}$ be the simplicial complex given in (\ref{main-cpx}).
If we take a quantum search driven by (\ref{qs}) with the initial state $\psi_{IN}$ given in (\ref{initial-state}), then there exists a time $t_f=O(\sqrt{N})$ 
such that $p_f\sim 1$. 
\end{theorem}
This theorem shows that the search algorithm on simplicial complexes (as a triangulation of $S^n$) finds the marked simplex $\tau_\ast$ with high probability independent of the system size $4N=2(n+1)(n+2)$, within the time $O(\sqrt{N}) = O(n)$.
Note that the marked simplex $\tau_\ast$ is uniquely determined by two $n$-simplices $|\sigma_\ast|, |\sigma_\ast'|\in \mathcal{K}$ in (\ref{main-cpx}) as the primary face $\tau_\ast = |\sigma_\ast|\cap |\sigma_\ast'|$.
Via the unitary equivalence, our problem is then translated into quantum search problems on graphs with $2(n+2)=O(n)$ vertices and $2(n+1)(n+2)=O(N)$ arcs for finding four marked arcs, which is the one choice among $O(n^2) = O(N)$ database entries.
Therefore our statement indicates that our simplicial search problem achieves the quantum speed-up like Grover's Algorithm, or quantum search over several graphs for finding marked vertices (\cite{ADFP2012, ADMP2010, P2013, S2008, SKW2003}).

\subsection{Proof of Theorem~\ref{thm-search}}
Here we give the proof of Theorem \ref{thm-search}, which is divided into three parts.
First note that, as we mentioned, the $n$-dimensional simplicial complex $\mathcal{K}$ in (\ref{main-cpx}) is strongly connected and orientable.
Moreover, $B_{n-1}(\mathcal{K}) = \emptyset$.
The induced bipartite graph $G_\cap(\mathcal{K}^{n,n-1})$ is therefore graph isomorphic, thanks to Lemma \ref{lem-ori}, to $S\circ D\circ G_a(\mathcal{K})$.

\subsubsection{Graph deformation induced by the quantum search}

Let $U_*=\hat{F}\hat{E}_*$ be the time evolution operator of the quantum search, 
where $\hat{E_*}=\oplus_{\tau\in\bra \mathcal{K}_{n-1} \ket}\hat{E}_\tau$ and $\hat{F}=\oplus_{\sigma\in\bra \mathcal{K}_{n} \ket}\hat{F}_\sigma$. 
Putting $U_0:=\hat{F}\hat{E}$ with 
	\[ \hat{E}:=\oplus_{\tau\in\bra \mathcal{K}_{n-1} \ket}\begin{bmatrix}0&1 \\ 1 &0\end{bmatrix}\] 
which is the time evolution operator without marked simplices, we have 
	\[ U_*=U_0 \bigoplus_{\tau \in \bra \mathcal{K}_{n-1} \ket} \hat{E}_\tau',   \]
where 
	\[\hat{E}_\tau'=\begin{cases} 
        \begin{bmatrix} 1 & 0 \\ 0 & 1 \end{bmatrix} & \text{: $|\tau|\neq \tau_*$,} \\ 
        \\
        -\begin{bmatrix} 0 & 1 \\ 1 & 0 \end{bmatrix} & \text{: $|\tau|= \tau_*$.}
        \end{cases}  \]
By Theorem~\ref{thm-equiv-duplicate}, it holds
	\[ \mathcal{W}U_*\mathcal{W}^{-1}=\hat{C}\hat{S}\hat{S}_*', \]
where $\hat{S}\hat{C}$ is the coined walk driven by the Grover operator on $D\circ G_a(\mathcal{K})$  
and $\hat{S}_*'=\mathcal{W}\oplus_{\tau}\hat{E}_\tau'\mathcal{W}^{-1}$. 
Recall that the vertex set of $D\circ G_a(\mathcal{K})$ is the disjoint union of $\mathcal{K}_n$ and its copy $\phi(\mathcal{K}_n)$, and 
the edge set is defined by $\{ |\sigma|\phi(|\sigma|') \;|\; |\sigma|\cap |\sigma|'\in \mathcal{K}_{n-1}\setminus\{\emptyset\}\}$. 
It holds that for any induced arc of $D\circ G_a(\mathcal{K})$,  
	\[ \hat{S}\hat{S}_*'\delta_a=\begin{cases}  -\delta_{a} & \text{: $|a|=|\sigma| \phi(|\sigma|')$ with $|\sigma|\cap|\sigma'|=\tau_*$} \\ 
        \delta_{\bar{a}} & \text{: $|a|=|\sigma| \phi(|\sigma|')$ with $|\sigma|\cap|\sigma'| \neq \tau_*$} \end{cases} \]
Here $|a|$ is the support edge of arc $a$. 
Therefore the operator $\hat{S}\hat{S}_*'$ makes a walker on the arc corresponding to $(\sigma,\tau)$ with $|\tau|=\tau_*$ 
stay at the same arc and change the phase, while a walker on the other arcs move to the inverse arc as is the usual shift operator. 
\par
From this observation, we deform the graph $D\circ G_a(\mathcal{K})$ to an directed graph by the following operation: 
the arc from $\phi(|\sigma|')$ to $|\sigma|$ where $\sigma,\sigma'\triangleright \tau$ with $|\tau|=\tau_\ast$ is rewired as the self loop of $|\sigma|$, 
as well as the arc from $|\sigma|$ to $\phi(|\sigma|')$ is rewired as the self loop of $\phi(|\sigma|')$. 
Here for the self loop $a$, we regard $t(a)=o(a)$. 
This deformed directed graph is denoted by $G_\ast$.
The illustration of $G_\ast$ is shown in Figure \ref{fig-dup-num} below.
Now it is easy to see that the induced coined walk $\Gamma_\ast$ is deformed as follows.
\begin{lemma}\label{reducedgraph}
Let $G_\ast =(V,A)$ be the above directed graph. Here $A$ is the set of (directed) arcs. 
The set of self loops are denoted by $\mathbb{S}_\ast$ 
The quantum search problem is reduced as follows: 
estimate the time $t_f$ at which we observe the subspace of $\ell^2(A)$; $\{\delta_a \;|\; a\in \mathbb{S}_\ast\}$ with a high probability.   
The time evolution $\hat{\Gamma}_\ast :=CS_\ast$ on $\ell^2(A)$ is expressed as follows:
\begin{align*}
C &= \bigoplus_{u\in V} \left( \frac{2}{n}J_n- I_n \right),\quad
(S_\ast\psi)(a) = \begin{cases} 
	\psi(\bar{a}) & \text{: $a\notin \mathbb{S}_\ast $,} \\ 
	-\psi(a) & \text{: $a\in \mathbb{S}_\ast $.} 
	\end{cases}
\end{align*}
\end{lemma}
%
%

\begin{remark}
The out-degree and in-degree of every vertex of $G_*$ are identically $n+1$ in this case. 
The vertex in $\phi(\mathcal{K}_n)$ which is not connected to $|\sigma|$ having the self loop is uniquely determined. 
This vertex denoted by $\phi(|\sigma|')$ has also a self loop. Moreover it holds that 
a vertex in $\mathcal{K}_n$ has a self loop iff its copy in $\phi(\mathcal{K}_n)$ has a self loop. 
Thus there are four vertices having the self loop $\{|\sigma|,\phi(|\sigma|),|\sigma|',\phi(|\sigma|') \}$ in $G_*$, which is independent of the system size. 
\end{remark}

{\rm
\begin{remark}
The above lemma is applicable to quantum search driven by SQW2 on any strongly connected and orientable simplicial complexes without boundary. 
\end{remark}
}

The above lemma claims that this simplicial quantum search converts the quantum search of arcs in graph $G_*$ driven by the coined walks. 
The reduced search problem on graph $G_*$ is explained as follows. 
Assume that we have the list which provides the adjacency relation of the graph $G_*$. 
From this list, we take a search of the self loop. 
Note that, in a so-called classical search, we need to take $(n+1)$ queries whether the terminus is same as the origin for each vertex. 
Thus we need $O(n^2)=O(N)$ queries for a classical search. 
\subsubsection{Spectral map}
We review the spectral mapping properties of quantum walks and associated self-adjoint operators based on arguments in \cite{MOSIIS}. 
Let $G_\ast=(V,A)$ be the induced graph with the coined walk $\hat{\Gamma}_\ast$ in Lemma \ref{reducedgraph}. 
Let $T_\ast$ be the self adjoint operator on $\ell^2(V)$ such that 
\begin{equation}
\label{T_ast}
        (T_\ast f)(u)= \sum_{a:t(a)=u}q(a)f(o(a)), 
\end{equation}
where $q: A\to\mathbb{R}$ is defined by 	
	\[ q(a)=(-1)^{\bs{1}_{\mathbb{S}_\ast}(a)}/n.  \]
We define $\partial^*_\theta: \ell^2(V)\to \ell^2(A)$ $(\theta\not \equiv 0 \mod\pi )$ by 
\begin{equation*}
(\partial^*_\theta f)(a)=\frac{1}{\sqrt{2n}|\sin \theta|} \times 
\begin{cases}
	f(o(a))-e^{i\theta}f(t(a)) & \text{: $a\notin \mathbb{S}_\ast$, } \\
	-(1+e^{i\theta})f(o(a)) & \text{: $a\in \mathbb{S}_\ast$. }
\end{cases}
\end{equation*}
Applying the spectral mapping property of quantum walks to $\hat \Gamma_\ast$, we have the following lemma.
\begin{lemma}\label{spectral}
Let $\hat \Gamma_\ast$ be the unitary operator appeared in Lemma \ref{reducedgraph}, and $\partial_\theta^\ast$ be the above. 
Then we have 
	\begin{align*}
        \sigma(\hat \Gamma_\ast) &= {\sf j}^{-1}(\sigma(T_\ast)) \cup \{\pm 1\},  \\
        \ker(\hat \Gamma_\ast-zI) &= \{ \partial^*_\theta f \;|\; f\in\ker(T_\ast-{\sf j}(z)I ) \} \;\mathrm{for}\;z/|z|\notin \{\pm 1\},
        \end{align*}
where ${\sf j}(z) = (z+z^{-1})/2$ for $z\in \mathbb{C}$. 
\end{lemma}
\begin{proof}
The proof follows from \cite{MOSIIS}. 
\end{proof}
Usual quantum search for coined walks is inherited by a random walk with Dirichlet boundary of the marked vertices. 
On the other hand, Lemma \ref{spectral} claims that simplicial quantum search reflects the spectral structure of the cellular automaton 
on which the associated moving weight on the self loop takes a negative value instead of boundary conditions. 

\subsubsection{Finding probability and time complexity}

We finish the proof of the theorem following the arguments by \cite{P2013} and its reference therein.
First, we obtain the following spectral information of $T_\ast$. 
The proof is shown in Appendix \ref{appendix-spectral}. 
\begin{lemma}
\label{lem-spectral}
Let $T_\ast$ be the symmetric matrix induced by $\mathcal{K}^{n,n-1}$ and determined by (\ref{T_ast}) 
whose size is $2(n+2) \times 2(n+2)$. 
Assume that the vertices having the self loop are labeled by $v_1$, $v_2$ and $v_{n+3}$, $v_{n+4}$ and their corresponding 
standard base in $\ell^2(V)$ are $[1,0,\dots,0]^{T}$, $[0,1,\dots,0]^{T}$, $[0,\dots,0,1,0,\dots,0]^{T}$ and $[0,\dots,0,0,1,\dots,0]^{T}$, respectively. 
The maximal eigenvalue and its eigenvector of $T_\ast$ are described by 
\begin{equation*}
\mu_1 = \frac{n-2+\sqrt{n(n+8)}}{2(n+1)}\;(<1), \quad f_1 =[\eta,\eta,\underbrace{1\dots,1}_{n},\eta,\eta,\underbrace{1\dots,1}_{n}]^T,
\end{equation*}
respectively, where
\begin{equation*}
\eta=\frac{1}{4}\left\{ -n+\sqrt{n(n+8)} \right\},\;\; \|f_1\|^2 = \frac{n}{2} \left\{ n+8-\sqrt{n(n+8)} \right\}. 
\end{equation*}	
\end{lemma}
Let $\alpha_\pm$ be the eigenvector inherited from the maximal eigenvector of $f_1$, that is, 
	\[ \alpha_\pm:= \partial^*_{\pm \arccos \mu_1} f_1. \]
We set $\beta_\pm:=(\alpha_+ \pm \alpha_-)/\sqrt{2}$. 
The following two lemmas hold. 
\begin{lemma}
\label{IN}
It holds that
\begin{equation*}
\langle \psi_{IN}, \beta_- \rangle=-\im+o(1).
\end{equation*}
\end{lemma}
\begin{proof}
See Appendix \ref{section-IN}.
\end{proof}
\begin{lemma}
\label{Tar}
Let $\psi_{Tar}\in \ell^2(A)$ be 
\begin{equation*}
\psi_{Tar}(a)=\bs{1}_{\mathbb{S}_\ast}(a)/2.
\end{equation*}
Then we have
\begin{equation*}
\langle \psi_{Tar}, \beta_+ \rangle=1+o(1).
\end{equation*}
\end{lemma}
\begin{proof}
See Appendix \ref{section-Tar}.
\end{proof}
The above lemmas indicate that, as Grover's algorithm for search problems (\cite{P2013}), the unitary evolution of the initial state $\psi_{IN}$ is essentially performed in the $2$-dimensional vector space $\spann \{\alpha_\pm \} = \spann \{\beta_\pm \}$ and that the search will find a marked simplex with probability close to $1$, which is due to the fact that $T_\ast$ does not have an eigenvalue $1$.
\par
\bigskip
Let $\beta_-\cong [1,0]^T$ and $\beta_+\cong [0,1]^T$.
The operator $\Gamma_\ast$ denotes the operator $\hat \Gamma_\ast$ restricted to the eigenspace $\spann\{\beta_\pm\}$.
Then it holds that
	\[ \Gamma_\ast \cong \begin{bmatrix} \cos\theta_1 & \im \sin\theta_1 \\ \im \sin\theta_1 & \cos\theta_1 \end{bmatrix}, \quad \theta_1 = \arccos \mu_1. \]
Thus when we choose the final time $t_f=[\pi/(2\theta_1)]$, then the $t_f$-th power of the representation matrix becomes 
	\[ \Gamma_\ast^{t_f} \cong \begin{bmatrix} 0 & \im \\ \im  & 0 \end{bmatrix}+o(1). \]
The final time $t_f$ is estimated by 
\begin{equation}
\label{time} 
        t_f\sim \frac{\pi}{2\theta_1} \sim \frac{\pi}{2\sin\theta_1} = \frac{\pi}{2\sqrt{1-\cos^2\theta_1}} = \frac{\pi}{2\sqrt{1-\mu_1^2}} \sim \frac{\pi}{4\sqrt{2}}n=O(n)=O(\sqrt{N}). 
\end{equation}
By Lemmas \ref{IN} and \ref{Tar}, we have 
	\begin{align*}
        (\Gamma_\ast)^{t_f} \psi_{IN} &= (\Gamma_\ast)^{t_f} \beta_- +o(1) \\
        &= \beta_+ +o(1) \\
        &= \psi_{Tar}+o(1)
        \end{align*}
Let $p'_f(\psi)$ be the overlap of $\psi\in \ell^2(A)$ to the state on marked simplex such that $p'_f(\psi)=|\bra \psi_{Tar}, \psi \ket|^2$ 
and $p_f(\psi)$ be the overlap to the subspace $\{\psi(a):a\in \mathbb{S}_\ast \}$, that is, $p_f(\psi)=\sum_{a\in \mathbb{S}_\ast}|\psi(a)|^2$.
We have $p_f'(\psi)\leq p_f(\psi)$ by the Schwartz's inequality. 
Then we have 
	\[ p_f(\Gamma^{t_f}\psi_{IN})\geq p_f'(\Gamma^{t_f}\psi_{IN})=1+o(1), \]
which completes the proof of Theorem~\ref{thm-search}. 
\subsection{Demonstrations}
We demonstrate simplicial quantum search with numerical simulations.
Fix the $n$-dimensional simplicial complex $\mathcal{K}$ as (\ref{main-cpx}) with various $n$, and see the behavior of simplicial quantum walks on $\mathcal{K}$ as well as finding probability of marked simplices, according to the following algorithm.
\par
\bigskip
Firstly, fix an integer $n\geq 2$.
Then construct the simplicial complex $\mathcal{K}=X(K_{n+2})^{(n)}$ and associated duplication graph $D\circ G_a(\mathcal{K})$.
The main focus for practical simulations is quantum walks on the graph $G_\ast$ deformed from $D\circ G_a(\mathcal{K})$.
Note that the directed graph $G_\ast = (V(G_\ast), A(G_\ast))$ has the following information:
\begin{align*}
&\sharp V(G_\ast) = 2(n+2),\quad \sharp A(G_\ast) = 2(n+1)(n+2)\equiv 4N,\\
&{\rm in\text{-}}\deg v = {\rm out\text{-}}\deg v \equiv n+1,\ \forall v\in V(G_\ast).
\end{align*}
\par
Secondly, fix a simplex $\tau_\ast \in \mathcal{K}_{n-1}$ as a marked simplex, which corresponds to the choice of two $n$-simplices $|\sigma_\ast|, |\sigma_\ast'|\in \mathcal{K}_n$, since each $(n-1)$-simplex in $\mathcal{K}$ is the primary face of uniquely determined two $n$-simplices in $\mathcal{K}$.
By the graph isomorphism of $G_\cap(\mathcal{K}^{n,n-1})$ to $D\circ G_a(\mathcal{K})$, the choice of $\tau_\ast$ corresponds to that of four  vertices corresponding to $|\sigma_\ast|$ and $|\sigma_\ast'|$ (with orientations), which is described in Figure \ref{fig-dup-num} as well as descriptions of $D\circ G_a(\mathcal{K})$ and $G_\ast$.

\begin{figure}[htbp]\em
\begin{minipage}{0.5\hsize}
\centering
\includegraphics[width=8.0cm]{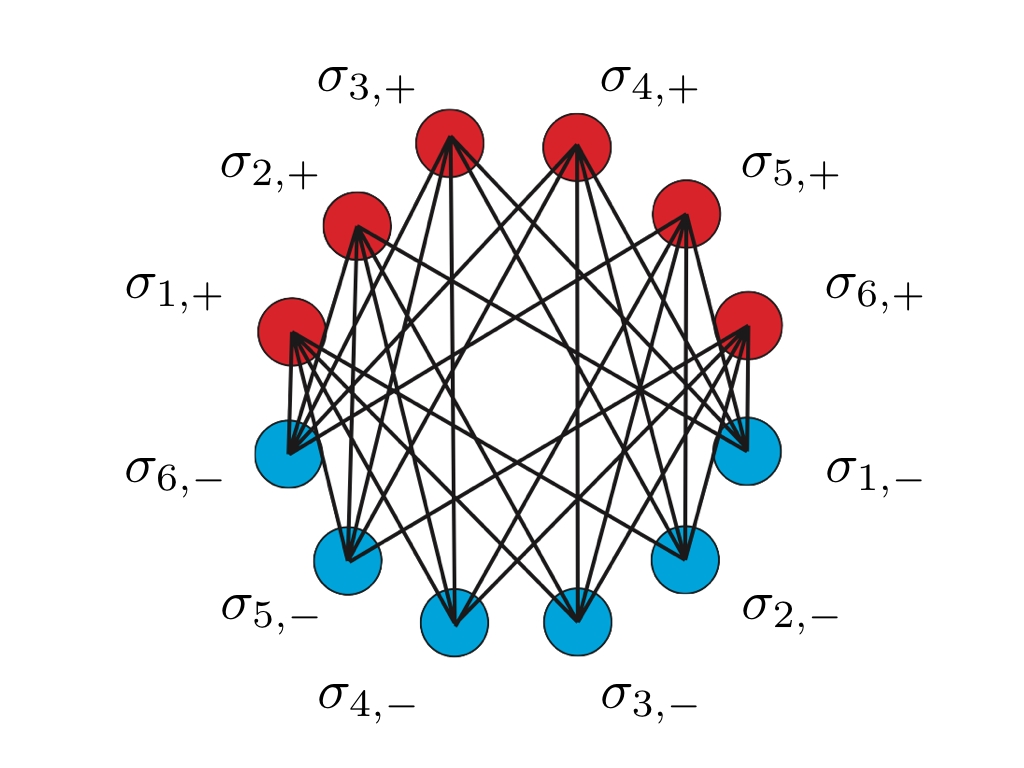}
(a)
\end{minipage}
\begin{minipage}{0.5\hsize}
\centering
\includegraphics[width=8.0cm]{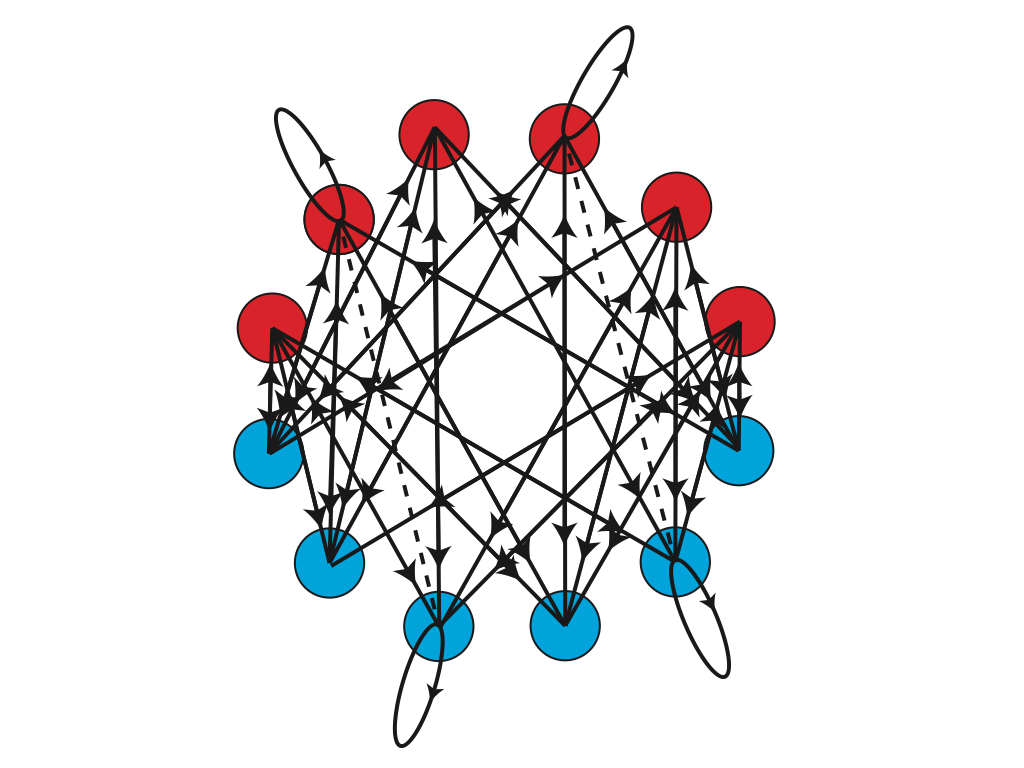}
(b)
\end{minipage}
\caption{Description of $D\circ G_a(\mathcal{K})$ associated with $\mathcal{K}$ in (\ref{main-cpx}) and $G_\ast$, $n=4$}
\label{fig-dup-num}
When $n=4$, then the $4$-dimensional simplicial complex $\mathcal{K}$ has six $4$-simplices, which corresponds to vertices of $D\circ G_a(\mathcal{K})$, six of which have non-contradicted orientations (red) and the others have the opposite ones (blue). 
Vertices are supposed to be uniformly distributed on the unit circle for visibility in Figure \ref{fig-search}.
The graph $D\circ G_a(\mathcal{K})$ is thus expressed as (a).
Now let the marked face $\tau_\ast$ be $\tau_\ast = |\sigma_2|\cap |\sigma_4|$, which is uniquely determined by the choice of two $n$-simplices.
Then cut the arcs connecting $\sigma_{2,\pm}$ and $\sigma_{4,\mp}$ and add self loops to these four vertices, which results in the construction of the deformed graph $G_\ast$ shown in (b).
Note that finding probability of walker at $\tau_\ast$ is attained as the sum of amplitudes on self loops.
\end{figure}

\par
Thirdly, set the initial state as the uniform state $\psi = \psi_{IN}$, where $\psi_{IN}$ is given in (\ref{initial-state}) with an identification with an element in $\ell^2(A(G_\ast))$.
Finally, compute the time evolution $\{\hat \Gamma^t_\ast \psi_{IN}\}_{t\geq 1}$.
Computation results are shown in Figures \ref{fig-search} - \ref{fig-time}, which follow the result described in Theorem \ref{thm-search}.

\begin{figure}[htbp]\em
\begin{minipage}{0.5\hsize}
\centering
\includegraphics[width=8.0cm]{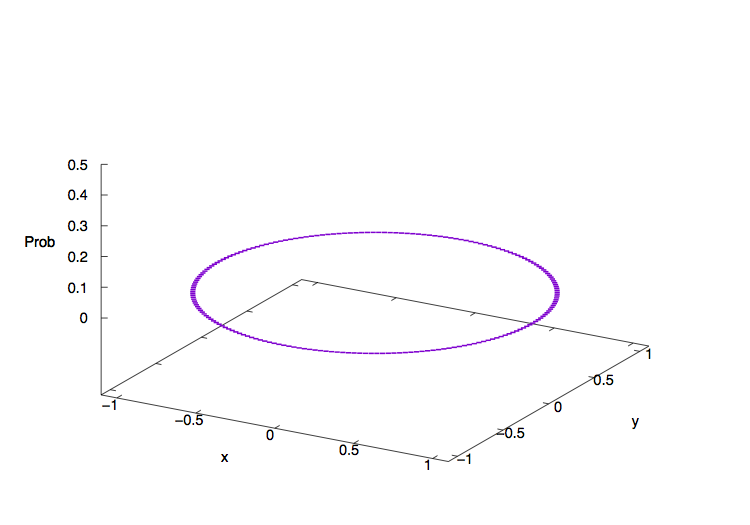}
(a)
\end{minipage}
\begin{minipage}{0.5\hsize}
\centering
\includegraphics[width=8.0cm]{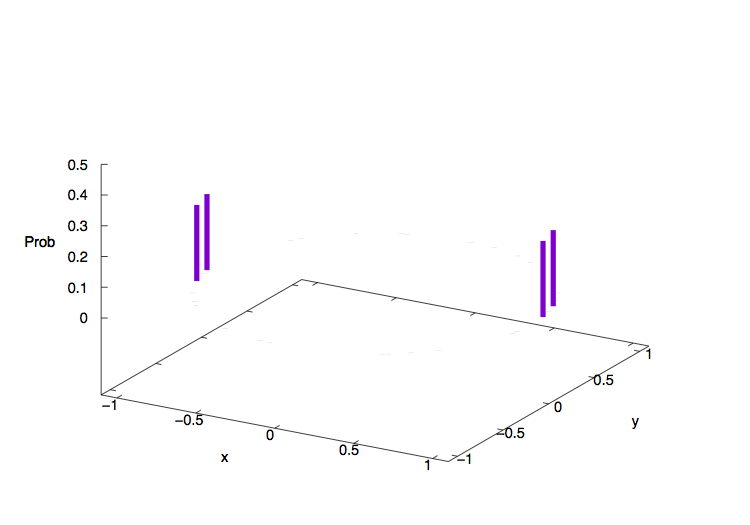}
(b)
\end{minipage}
\caption{Simplicial quantum search}
\label{fig-search}
(a) : Initial (uniform) state $\psi_{IN}$ of quantum search on $X(K_{n+2})^{(n)}$ consisting of $n$-simplices $\{|\sigma_i|\}_{i=1}^{n+2}$. 
Distribution of vertices is followed by Figure \ref{fig-dup-num}.
In these figures, $n+2$ is set to be $100$.
The marked simplex is $\tau_\ast = |\sigma_5|\cap |\sigma_{11}|$, which corresponds to four vertices on associated deformed graph $G_\ast$.
(b) : The evolved state $\hat \Gamma^t \psi_{IN}$ at $t = t_f \approx 55$.
Four peaks are observed at vertices corresponding to $\sigma_{5,\pm}$ and $\sigma_{11,\pm}$.
Finding probability at each vertex is close to $0.25$, and hence $p_f$ is close to 1.
This observation means that the marked simplex $\tau_\ast$ is found with probability $p_f\sim 1$ after the time $t_f$.
\end{figure}

\begin{figure}[htbp]\em
\begin{minipage}{0.5\hsize}
\centering
\includegraphics[width=8.0cm]{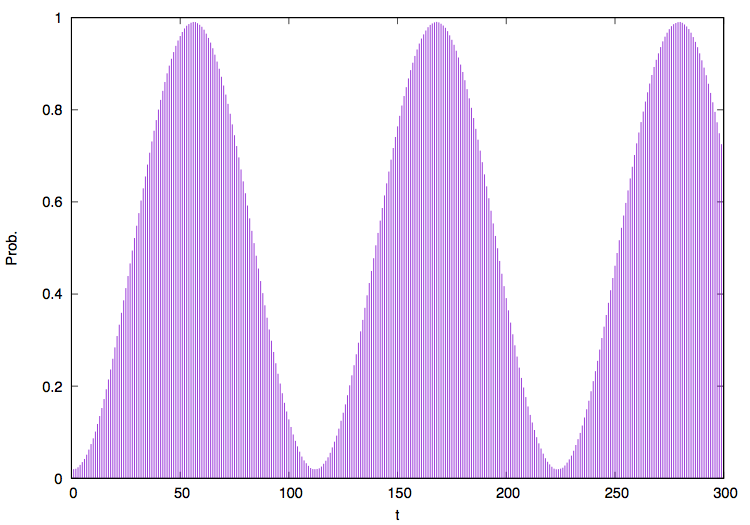}
(a)
\end{minipage}
\begin{minipage}{0.5\hsize}
\centering
\includegraphics[width=8.0cm]{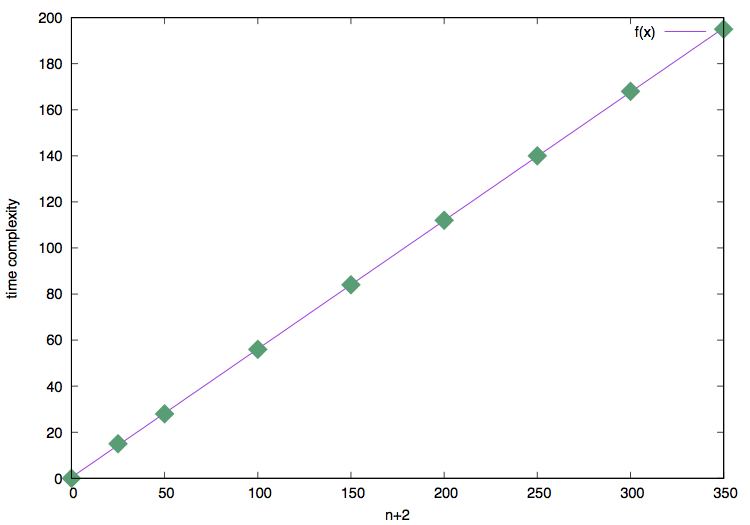}
(b)
\end{minipage}
\caption{Pseudo-periodicity and time complexity of simplicial quantum search}
\label{fig-time}
(a) : Finding probability $p_f$ of $\tau_\ast$ corresponding to four vertices with self loops drawn in Figure \ref{fig-search}-(b). 
The value $p_f$ periodically oscillates between $\approx 0.005$ and $\approx 0.99$.
(b) : plot of time complexity $t_f$ for $n+2= 50, 100, 150, 200, 250, 300, 350$ and the fitting function $f(x) = 0.5574x + 0.411667$ obtained by gnuplot. 
We can observe that $0.5574\approx \pi/(4\sqrt{2})$, which is the coefficient of the first order of the complexity in (\ref{time}). 
This graph indicates that $t_f = O(n)$, equivalently, $t_f = O(\sqrt{N})$ where $4N=2(n+1)(n+2)$.
\end{figure}

\section{Conclusion and further directions}
In this paper, we have discussed quantum search on simplicial complexes.
To this end, we have introduced simplicial quantum walks as an alternative to the original simplicial quantum walks introduced in \cite{MOSver1}.
The alternative simplicial quantum walks are unitary equivalent to bipartite walks on induced bipartite graphs.
Moreover, under geometric constraints on simplicial complexes, simplicial quantum walks driven by Grover-like operators are also equivalent to coined quantum walks on graphs with duplication structure.
These unitary equivalence yield the applicability of quantum search on graphs to search problems on simplicial complexes.
Indeed, we have proved that, via the unitary equivalence, the search algorithm on the $n$-dimensional simplicial complex $X(K_{n+2})^{(n)}$ introduced in Section \ref{section-search} finds a marked simplex with finding probability $p_f\sim 1$, which is independent of the dimension $n$, with the time complexity $O(n)=O(\sqrt{N})$. 
From the viewpoint of quantum search on graphs, the simplicial quantum search in our setting is regarded as the search of {\em four self loops} on the graphs simultaneously (not only either of them), which corresponds to a problem for the choice of one pairs of associated vertices among $O(N) = O(n^2)$ possibilities. 
Therefore our result gives an example that our equipments achieve quantum speed-up in searching problem over simplicial complexes.
\par
The result also reveals an insight to dynamics of coined quantum walks on graphs.
Namely, our simplicial quantum search problem turns out to be unitary equivalent to simultaneous search of two marked arcs on regular graphs.
We believe that the current study will open the door for connecting quantum search problems on graphs and topology.
\par
We end the paper leaving several future directions.

\subsection{Quantum walks and search problems on other complexes}
In this paper, we have only considered search problems on the simplicial complex $X(K_{n+2})^{(n)}$, $n\geq 2$.
As a sequel, it is thus natural to consider the dynamics and search problems on the following increasing sequence (namely, filtration) of skeletons:
\begin{align*}
K_{n+2} &= X(K_{n+2})^{(1)} \subset X(K_{n+2})^{(2)} \subset \cdots \\
	&\subset X(K_{n+2})^{(n-1)} \subset X(K_{n+2})^{(n)}\subset X(K_{n+2})^{(n+1)} = \text{(an $(n+1)$-simplex)}.
\end{align*}
In particular, the cases $\{X(K_{n+2})^{(\kappa)}\}_{\kappa = 2,\cdots, n-1}$ remain nontrivial.
In contrast to $X(K_{n+2})^{(n)}$, simplicial complexes $X(K_{n+2})^{(\kappa)}$ with $\kappa = 2,\cdots, n-1$ are not orientable since each primary face $\tau$ has junctions; namely, $\tau$ has more than two cofaces.
In particular, we cannot apply Theorem \ref{thm-equiv-duplicate} to these complexes for solving search problems.

\par
From the viewpoint of differential geometry, it is natural to consider simplicial complexes as triangulations of differentiable manifolds other than the unit sphere $S^n$, such as cylinder, torus (not $n$-dimensional lattices with periodic boundary conditions), M\"{o}bius band, Klein bottle and so on.
There are a lot of manifolds which are compact with no boundaries (namely, closed) and orientable. 
Theorem \ref{thm-equiv-duplicate} can be still applied to such manifolds, 
which indicates that exploring how these geometrical changes deform our result of Theorem~\ref{thm-search} is one of the interesting future's problems.
As for non-orientable manifolds like M\"{o}bius band, as in cases $\{X(K_{n+2})^{(\kappa)}\}_{\kappa = 2,\cdots, n-1}$,  the other feature of associated graph $S\circ D\circ G_{a}(\mathcal{K}) \bullet B_{n-1}$ should be concerned.

\subsection{Compatibility with ordinary coined walks}
SQW2 constructed in the present arguments is supposed to be defined on $n$-dimensional simplicial complexes with $n\geq 2$.
When $n=1$, as indicated in Remark \ref{rem-dim1}, coined walks are not always restored from SQW2.
SQW2 on graphs actually define coined walks on {\em double graphs}, in which sense the problem whether SQW2 is a natural extension of coined walks on graphs remains open.
In other words, it is a non-trivial problem when the double structure of SQW2 on graphs generates a difference from coined walks on the same graphs.
We leave a definition of simplicial quantum walks as a natural extension of coined walks on graphs which are efficient to study from both mathematical and computational viewpoints to be considered in future works.

\section*{Acknowledgments}
KM was partially supported by Program for Promoting the reform of national universities (Kyushu University), Ministry of Education, Culture, Sports, Science and Technology (MEXT), Japan, World Premier International Research Center Initiative (WPI), MEXT, Japan, and JSPS Grant-in-Aid for Young Scientists (B) (No. JP17K14235).
OO was partially supported by JSPS KAKENHI Grant (No. JP24540208, JP16K05227). 
ES acknowledges financial support from the Grant-in-Aid for Young Scientists (B) and of Scientific Research (B) Japan Society for the Promotion of Science (Grant No.~JP16K17637, No.~JP16K03939).
Finally authors would thank to reviewers for providing them helpful comments about contents of the paper.

\appendix
\section{Simplicial complexes : quick review}
\label{appendix-cpx}

We state a quick review of simplicial complexes for readers who are not familiar with them. See e.g. \cite{Z2005} for details.

\subsection{Definition}
\begin{definition}\rm
Let $\mathbb{R}^N$ be a Euclidean space and $O_N$ be the origin of $\mathbb{R}^N$. Let $a_0,a_1,\cdots, a_n \in \mathbb{R}^N$ be points so that $n$ vectors $\{\overrightarrow{a_0a_i}\}_{i=1}^n$ are linearly independent.
An {\em $n$-simplex} is a set $|\sigma| \subset \mathbb{R}^N$ given by
\begin{equation*}
|\sigma| = \left\{\sum_{i=1}^n \lambda_i \overrightarrow{a_0a_i}\mid \lambda_i \geq 0,\ \sum_{i=1}^n \lambda_i = 1\right\}.
\end{equation*}
We also write $|\sigma|$ as $|a_0a_1 \cdots a_n|$ if we write the dependence of points $\{a_i\}_{i=0}^n$ explicitly.
A $k$-{\em face} of an $n$-simplex $|\sigma| = |a_0a_1\cdots a_n|$ is a $k$-simplex $|\tau|$ generated by $k$ points in $\{a_0\}_{i=0}^n$. In such a case, $|\sigma|$ is called {\em a coface} of $|\tau|$.
An $(n-1)$-face of an $n$-simplex $\sigma$ is often called {\em a primary face} of $\sigma$.
\end{definition}
For example, for a given simplex $|\sigma| = |abc|$, edges $|ab|$, $|bc|$ and $|ca|$ are primary faces of $\sigma$. Also, vertices $|a|$, $|b|$ and $|c|$ are $0$-faces of $|\sigma|$. Finally, $|\sigma|$ is a coface of $|a|$, $|b|$, $|c|$, $|ab|$, $|bc|$ and $|ca|$.

\begin{definition}\rm
A {\em simplicial complex} $\mathcal{K}$ is the collection of simplices satisfying
\begin{itemize}
\item If $|\sigma| \in \mathcal{K}$, then all faces of $|\sigma|$ are also elements in $\mathcal{K}$.
\item If $|\sigma_1|, |\sigma_2| \in \mathcal{K}$ and if $|\sigma_1|\cap |\sigma_2| \not = \emptyset$, then $|\sigma_1|\cap |\sigma_2|$ is a face of both $|\sigma_1|$ and $|\sigma_2|$. 
\end{itemize}
For a given simplicial complex $\mathcal{K}$, the union of all simplices of $\mathcal{K}$ is the {\em polytope} of $\mathcal{K}$ and is  denoted by $\mathcal{K}$. A set $P$ is a {\em polyhedron} if it is the polytope $\mathcal{K}$ of a simplicial complex $\mathcal{K}$.
\end{definition}

Let $\mathcal{K} = \{\mathcal{K}_k\}_{k\geq 0}$ be a simplicial complex, where $\mathcal{K}_k = \{|\sigma| \in \mathcal{K} \mid |\sigma| \text{ is a $k$-simplex}\}$. If $n = \max \{k \mid \mathcal{K}_k \not = \emptyset\} < \infty$, then we call $\mathcal{K}$ an {\em $n$-dimensional} simplicial complex. 

\begin{definition}\rm
\label{dfn-strong-conn}
For a simplicial complex $\mathcal{K}$, a {\em facet} in $\mathcal{K}$ is a simplex $\sigma \in \mathcal{K}$ which is maximal with respect to the inclusion relation of sets. A simplicial complex $\mathcal{K}$ is {\em pure} if all facets in $\mathcal{K}$ have an identical dimension. An $n$-dimensional pure simplicial complex $\mathcal{K}$ is {\em strongly connected} if, for each $|\sigma|, |\tau| \in \mathcal{K}_n$, there is a sequence of $n$-simplices $\{|\sigma_j|\}_{j=0}^k$ with $|\sigma_0| = |\sigma|$ and $|\sigma_k| = |\tau|$ such that $|\sigma_{j-1}| \cap |\sigma_j|$ is a primary face of $|\sigma_{j-1}|$ and $|\sigma_j|$ for $j=1,\cdots, n$. 
\end{definition}

\begin{remark}\rm
We often call an $(n-1)$-face of an $n$-simplex $|\sigma|$ {\em a facet of $|\sigma|$}. It is completely different from facets in simplicial complexes.
\end{remark}
For a given simplicial complex $\mathcal{K}$, we can consider the sub-collection of simplices in $\mathcal{K}$ which itself is also a simplicial complex.
More precisely, let $\mathcal{L}$ be the following collection:
\begin{equation*}
\mathcal{L} := \{\sigma' \text{: simplex}\mid \exists \sigma\in \mathcal{K} \text{ such that }\sigma'\subset \sigma\} \subset \mathcal{K}.
\end{equation*}
If $\mathcal{L}$ itself is a simplicial complex, then $\mathcal{L}$ is called a {\em subcomplex} of $\mathcal{K}$.
In our main discussion, we consider the following type of subcomplexes.

\begin{definition}[Skeleton]\rm
Let $\mathcal{K}$ be an $n$-dimensional simplicial complex.
For $m=\{0,\cdots, n\}$, define a collection $\mathcal{K}^{(m)}$ by
\begin{equation*}
\mathcal{K}^{(m)} := \{\sigma \in \mathcal{K} \mid \dim \sigma \leq m\}.
\end{equation*}
It easily follows that $\mathcal{K}^{(m)}$ is a subcomplex of $\mathcal{K}$.
The complex $\mathcal{K}^{(m)}$ is called the {\em $m$-skeleton} of $\mathcal{K}$.
Obviously $\mathcal{K}^{(n)} = \mathcal{K}$ holds.
The $1$-skeleton $\mathcal{K}^{(1)}$ is a graph embedded in $\mathcal{K}$.
\end{definition}

\subsection{Orientability}
\label{section-orientability}
Let $\mathcal{K}$ be an $n$-dimensional strongly connected simplicial complex. 
Each $k$-simplex $|\sigma| = |v_0 \cdots v_k|$ determines {\em an orientation} as a choice of equivalent class determined by the equivalence relation $\sim_k$ (see Section \ref{section-setting} for details).
In particular, each simplex admits two orientations.
For an oriented simplex $\sigma$, the corresponding simplex $|\sigma|$ ignoring the orientation (namely, the order of vertices) is often called {\em the support } of $\sigma$.
For $k=0,1,\cdots, n$, let $\langle \mathcal{K}_k \rangle$ be the set of oriented $k$-simplices in $\mathcal{K}$; namely, each $k$-simplex $\sigma\in \langle \mathcal{K}_k \rangle$ is distinguished not only the support $|\sigma|$ but also orientations on $|\sigma|$.
It is then natural to consider the relationship of orientations on adjacent simplices, which induces a discussion of {\em global} orientability of simplicial complexes.
See Section \ref{section-setting} for details of $\mathcal{K}^{n,n-1}$.

\begin{definition}\rm
We say a pair of $n$-simplices $\sigma$ and $\sigma'$ is {\em non-contradicted} if and only if 
\begin{equation*}
 |\sigma|\neq |\sigma'| \mathrm{\;and\;} \exists \tau,\tau'\in \bra \mathcal{K}_{n-1} \ket 
\mathrm{\;with\;|\tau|=|\tau'|,\;\tau\neq \tau'\;}
\mathrm{\;s.t.,\;} 
(\sigma,\tau),(\sigma',\tau')\in \mathcal{K}^{n,n-1},
\end{equation*}
namely, $\sigma$ and $\sigma'$ has a common primary face $|\tau|$ whose induced orientation is opposite to each other.
\end{definition}
The above definition need no requirements for $(n-1)$-simplices with just one coface.
The collection of such $(n-1)$-simplices is referred to as {\em the boundary} $B_{n-1} = B_{n-1}(\mathcal{K})$ of $\mathcal{K}$.
Remark that each $(n-1)$-simplex has a coface since $\mathcal{K}$ is strongly connected, in particular pure.
\par
Since there exist two elements of $\bra \mathcal{K}_n \ket$ whose supports are commonly $|\sigma|$, 
we take the labeling to each element of $\bra \mathcal{K}_n \ket$ by $(|\sigma|,\epsilon)$, where $\epsilon\in \{\pm \}$. 
Remark that the labeling way of $\{\pm \}$ to each element of $\bra \mathcal{K}_n \ket$ is arbitrary in the present stage; indeed there are $2^{|\mathcal{K}_n|}$ choices. 
If the choice of orientations on all simplices are non-contradicted, then the {\em orientation on simplicial complexes} makes sense.
More precisely, the (global) orientation of simplicial complexes is defined as follows.

\begin{definition}\rm
Let $\mathcal{K}$ be an $n$-dimensional strongly connected simplicial complex. 
If there exists a sequence of orientations of $n$-simplices $(\epsilon_{|\sigma|})_{|\sigma|\in \mathcal{K}_n}\in \{\pm \}^{\mathcal{K}_n}$ such that all pairs
$\{ (|\sigma|,\epsilon_{|\sigma|}), (|\sigma|',\epsilon_{|\sigma|'}') \}$ with $|\sigma|\cap |\sigma|'\in \mathcal{K}_{n-1}\setminus \{\emptyset\}$ are non-contradicted, then we call the simplicial complex $\mathcal{K}$ {\em orientable}.
\end{definition}

\begin{figure}[htbp]\em
\begin{minipage}{0.32\hsize}
\centering
\includegraphics[width=6.0cm]{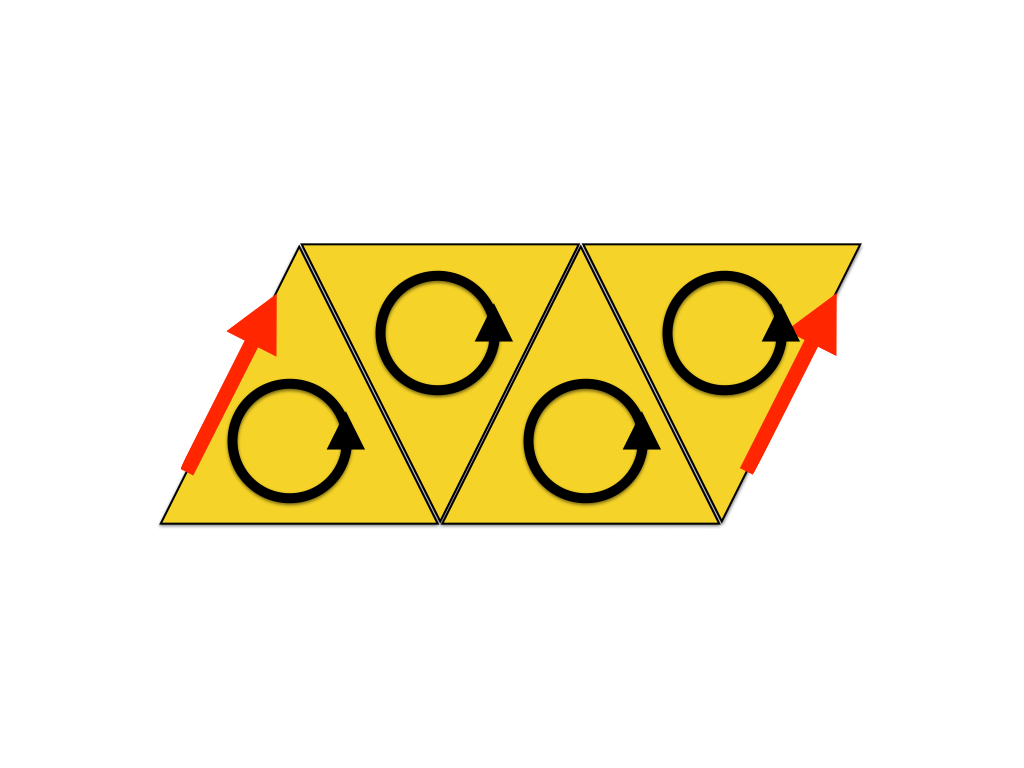}
(a)
\end{minipage}
\begin{minipage}{0.32\hsize}
\centering
\includegraphics[width=6.0cm]{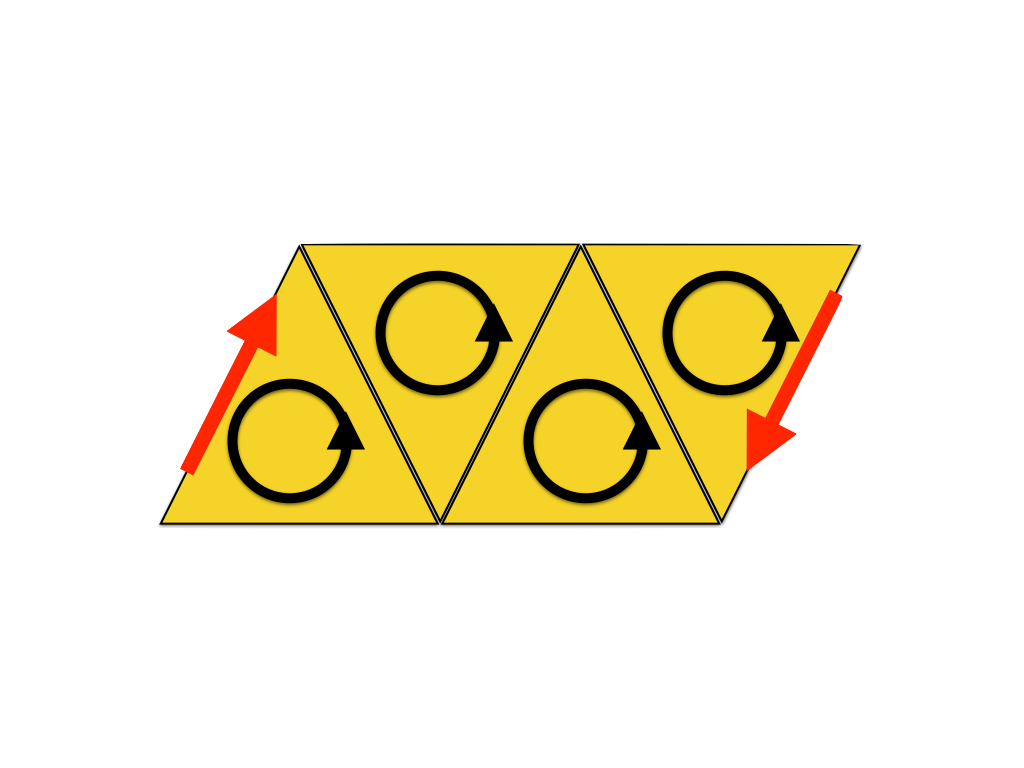}
(b)
\end{minipage}
\begin{minipage}{0.32\hsize}
\centering
\includegraphics[width=6.0cm]{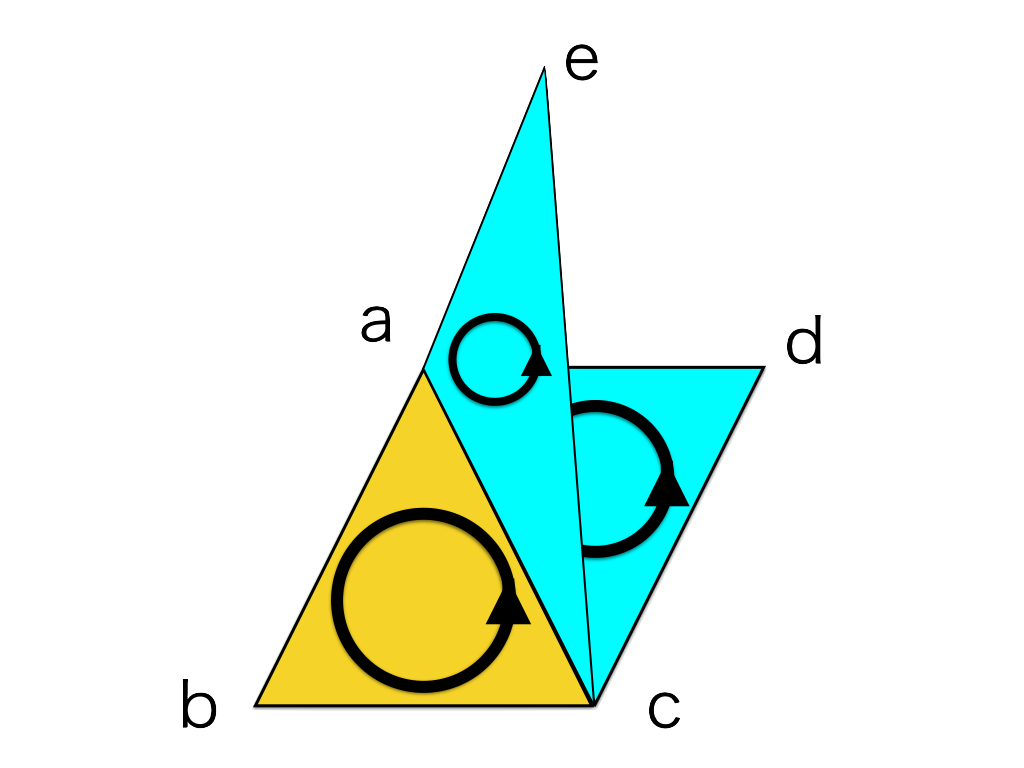}
(c)
\end{minipage}
\caption{Orientation of simplicial complexes}
\label{fig-orientation}
(a) : Two arrows mean the identification of edges with orientation. 
In this case, orientations of all adjacent simplices are non-contradicted, and hence the complex is orientable.
(b) : Unlike the case (a), the adjacent simplices including arrows induce the same orientation on the edge with arrow.
It can be the case for any adjacent simplices, and hence the complex is not orientable.
(c) : The complex has a junction at $|ac|$.
For a given orientation on $|abc|$, the non-contradicted orientation would be set on $|acd|$ and $|ace|$ shown in (c). 
However, such orientations induce the identical orientation on $|ac|$ both from $|acd|$ and $|ace|$.
This is always the case in the presence of junction like $|ac|$, and hence the complex is not orientable.
\end{figure}

Orientability of simplicial complexes is illustrated in Figure \ref{fig-orientation}.
If $\mathcal{K}$ is orientable and $(\epsilon_{|\sigma|})_{|\sigma|\in \mathcal{K}_n}$ gives the non-contradiction, then 
$(-\epsilon_{|\sigma|})_{|\sigma|\in \mathcal{K}_n}$ also automatically gives the non-contradictoriness.
On the other hand, such the choice only gives the possibility of the sequence justifying the orientability of $\mathcal{K}$.
We say a choice of one sequence from $(\pm\epsilon_{|\sigma|})_{|\sigma|\in \mathcal{K}_n}$ {\em an orientation on $\mathcal{K}$}.
Therefore, if $\mathcal{K}$ is orientable, then it has just two orientations.
The consequence of orientability for $\mathcal{K}$ is an analogue of orientability of differentiable manifolds.

\subsection{Clique complexes}
Typical simplicial complexes are often constructed by triangulation of differentiable manifolds, as shown in \cite{MOSver1}.
On the other hand, there is a way to construct simplicial complexes from given graphs.
Such complexes are called {\em clique complexes} of graphs.
More precisely, they are defined as follows.
\begin{definition}[Clique complex]\rm
Let $G=(V,E)$ be a graph.
Then construct a simplex $\sigma$ associated with $G$ as the following one-to-one correspondence:
\begin{align*}
&K_{m+1} = (\{v_i\}_{i=0}^m)\subset G : \text{ complete subgraph of $G$ with $m+1$ vertices in $V$}\\
&\quad \Leftrightarrow \quad \sigma = |v_0\cdots v_m| : \text{ an $m$-simplex}.
\end{align*}
Since the complete graph $K_{m+1}$ contains complete subgraphs $K_l$ for $l=2,\cdots, m$, then it easily follows that the collection
of simplices $X(G) := \{\sigma \mid \text{ made by the above correspondence }\}$ has a structure of simplicial complex.
We call the simplicial complex $X(G)$ {\em the clique complex} of $G$.
\end{definition}
An example of clique complexes is shown in Figure \ref{fig-clique}.

\begin{figure}[htbp]\em
\begin{minipage}{0.5\hsize}
\centering
\includegraphics[width=5.0cm]{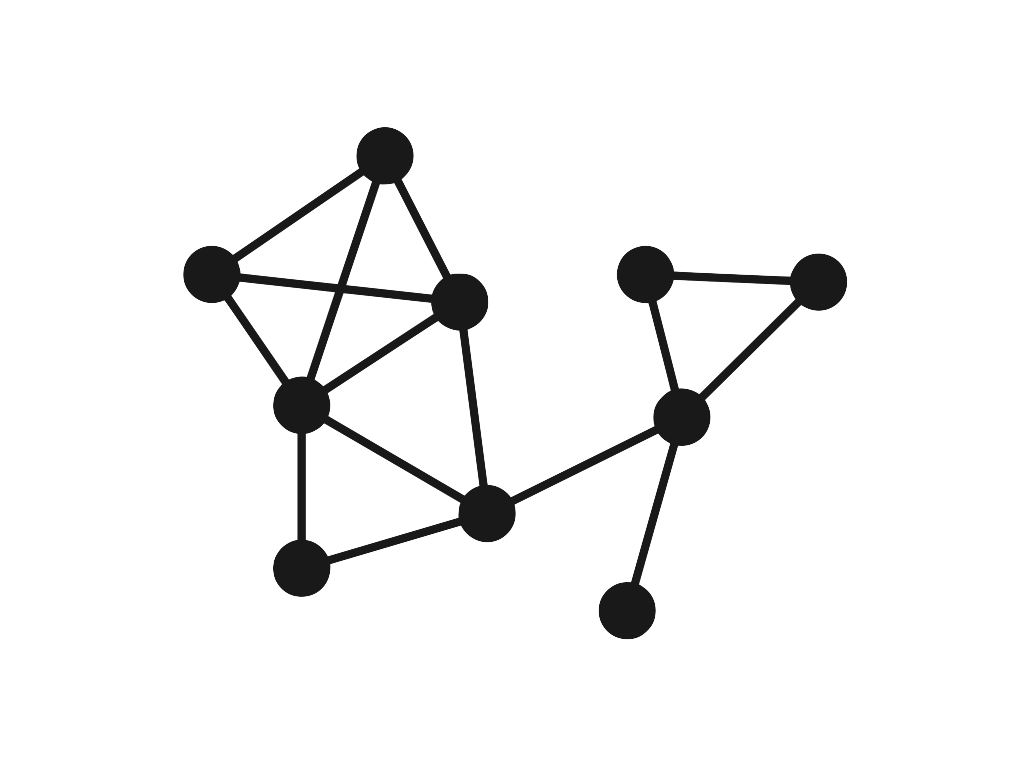}
(a)
\end{minipage}
\begin{minipage}{0.5\hsize}
\centering
\includegraphics[width=5.0cm]{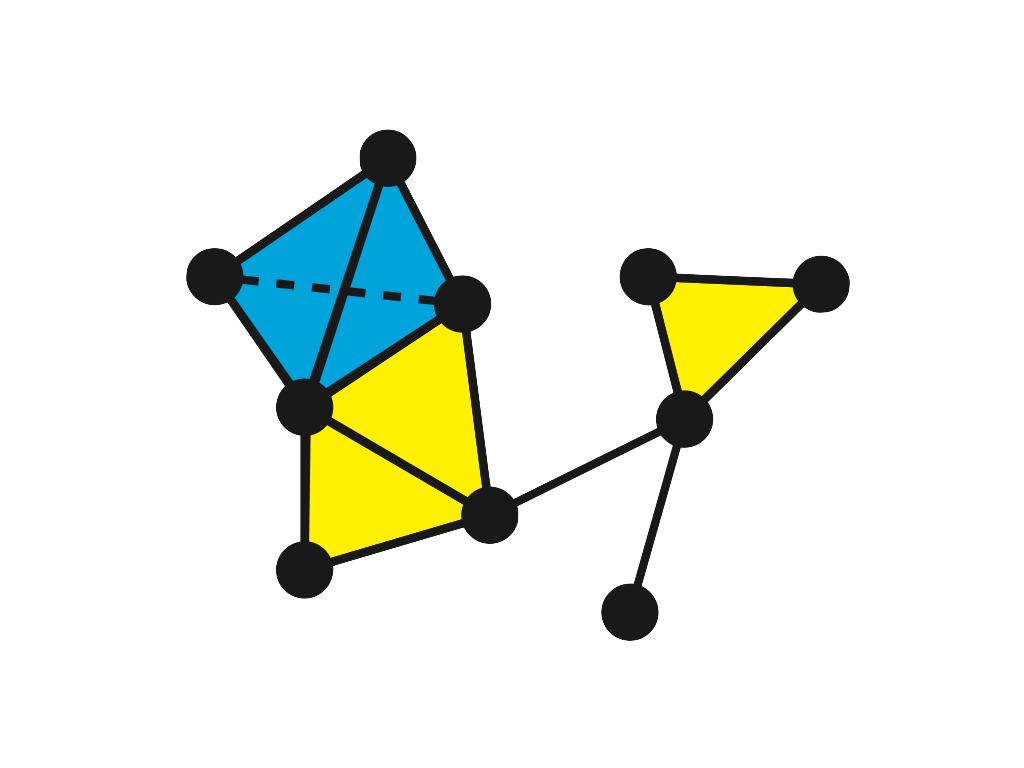}
(b)
\end{minipage}
\caption{A graph and its clique complex}
\label{fig-clique}
(a) : a graph $G$. This graph consists of one $K_4$, seven $K_3$, fifteen $K_2$ (edges) and ten $K_1$ (vertices).
(b) : the clique complex $X(G)$ of $G$, which consists of one $3$-simplex, seven $2$-simplices, fifteen $1$-simplices and ten $0$-simplices.
\end{figure}

\section{Spectral analysis of $T_\ast$}
\label{appendix-spectral}
Here we give detailed proofs of Lemmas \ref{lem-spectral}, \ref{IN} and \ref{Tar}, which are followed from standard arguments in quantum search problems (e.g., \cite{P2013} or references). 

\subsection{Proof of Lemma \ref{lem-spectral}}
We put $T_\ast-\lambda I=: M_\lambda$ induced by $\mathcal{K}^{n,n-1}$. The matrix size is $2(n+2)\times 2(n+2)$. 
Recall that $G_\ast$ is a bipartite graph with four self loops. 
We set the vertex set $X\sqcup Y$ of $G_\ast$ by $X=\{x_1,x_2,\dots,x_{n+2}\}$ and $Y=\{y_1,y_2,\dots,y_{n+2}\}$,
where, without the loss of generality via the change of base elements, $x_1, x_2$ and $y_1,y_2$ are assumed to have self loops. 
We fix the computational basis of $T_\ast$ is $\{\delta_{x_1},\dots, \delta_{x_{n+2}},\delta_{y_1},\dots, \delta_{y_{n+2}}\}$ by this order. 
Thus the matrix expression of $M_\lambda$ is described as follows: 
\begin{equation*}
M_\lambda = \begin{bmatrix} A_\lambda & B \\ B & A_\lambda \end{bmatrix},
\end{equation*}
where $A_\lambda$ and $B$ are $(n+2)\times (n+2)$ matrices such that 
\begin{equation*}
A_\lambda=\begin{bmatrix} 
        \alpha &        &       &        &        \\
               & \alpha &       &        &        \\
               &        & \beta &        &        \\
               &        &       & \ddots &        \\
               &        &       &        & \beta 
        \end{bmatrix}, \quad
B= \frac{1}{n+1}
\begin{bmatrix} 
        0 & 0 & 1 & 1 & \cdots & 1  \\ 
        0 & 0 & 1 & 1 & \cdots & 1  \\ 
        1 & 1 & 0 & 1 & \cdots & 1  \\ 
        1 & 1 & 1 & 0 & \cdots & 1  \\
        \vdots & \vdots & \vdots & \vdots & \ddots & \vdots  \\
        1 & 1 & 1 & 1 & \cdots  & 0 
\end{bmatrix} 
\end{equation*}
with 
\begin{equation*}
\alpha=-\frac{1}{n+1}-\lambda,\quad \beta=-\lambda.
\end{equation*}
From now on, we assume $\alpha,\beta\neq 0$, that is, $\lambda\neq -1/(n+1),0$, respectively.
Letting $f,g\in\mathbb{C}^{n+2}$, under this assumption, we have 
\begin{align}\label{f}
        \begin{bmatrix} f \\ g \end{bmatrix}\in \ker(M_\lambda) \notag 
        	& \Leftrightarrow \begin{bmatrix} f \\ g \end{bmatrix} \in \ker \begin{bmatrix} I & A_\lambda^{-1}B \\ 0 & I-(A_\lambda^{-1}B)^2 \end{bmatrix} \notag \\
                & \Leftrightarrow f=-A_\lambda^{-1}Bg,\;g\in \ker(I-A_\lambda^{-1}B)+\ker(I+A_\lambda^{-1}B).
\end{align}
The matrix $I+A_\lambda^{-1}B$ is expressed by
\begin{equation*}
 I+A_\lambda^{-1}B = \begin{bmatrix}  I_2 & R \\ L & M_\lambda'' \end{bmatrix},
\end{equation*}
where
\begin{align*} 
        R &=\frac{\alpha^{-1}}{n+1}\begin{bmatrix} 1 & \cdots & 1 \\ 1 & \cdots & 1 \end{bmatrix},\;
        L  = \frac{\beta^{-1}}{n+1}\begin{bmatrix} 1 & 1 \\  \vdots & \vdots \\ 1 & 1 \end{bmatrix}, \\
        M_\lambda''&=\frac{\beta^{-1}}{n+1}J_n +\left( 1-\frac{\beta^{-1}}{n+1} \right)I_n,
\end{align*}
$I_2$ is the $2$-dimensional identity matrix and $J_n$ is the $n$-dimensional all $1$ matrix. 
Letting $h_1\in \mathbb{C}^2$ and $h_2\in \mathbb{C}^n$, we have 
\begin{align}
\label{g}
g\in \ker(I+A_\lambda^{-1}B) \notag 
        	& \Leftrightarrow \begin{bmatrix} h_1 \\ h_2 \end{bmatrix} \in \ker \begin{bmatrix} I_2 & R \\ 0 & LR-M_\lambda'' \end{bmatrix} \notag \\
                & \Leftrightarrow h_1=-Rh_2,\;h_2\in \ker(LR-M_\lambda'')
        \end{align}
The matrix $LR-M_\lambda''$ is expressed by 
	\[ LR-M_\lambda''= x J_n + yI_n, \]
where 
	\[ x=\frac{\lambda+3/(n+1)}{\lambda+1/(n+1)}\; \frac{1}{(n+1)\lambda},\;y=-\frac{\lambda+1/(n+1)}{\lambda}. \] 
If $x\neq 0$, that is, $\lambda\neq -3/(n+1)$, then 
\begin{align}\label{h}
        h_2\in \ker(LR-M_\lambda'') \notag \ 
        	&\Leftrightarrow J_n h_2=-y/x h_2 \notag \\
                &\Leftrightarrow -y/x\in \{n,0\}, \notag \\
                &\qquad\qquad h_2\in \spann \{ \begin{bmatrix} 1 & \cdots & 1 \end{bmatrix}^T \}+ \spann \{ \begin{bmatrix} 1 & \cdots & 1 \end{bmatrix}^T \}^{\perp} \notag \\
                &\Leftrightarrow \lambda\in \{\lambda_{\pm}\}, \;h_2= \begin{bmatrix} 1 & \cdots & 1 \end{bmatrix}^T, 
\end{align}
where 
	\[\lambda_{\pm}:=\frac{n-2\pm \sqrt{n^2+8n}}{2(n+1)}\]
Since $-y/x=0$ if and only if $\lambda=-1/(n+1)$, the third equivalence holds under the assumption $\alpha\neq 0$.  
Therefore the spectrum of $T_\ast$ includes $\lambda_{\pm}$ $(\lambda_+\geq \lambda_-)$, and 
the candidates of the eigenvalues are $-3/(n+1)$ and $-1/(n+1)$ in the present stage. 
We consider $\ker(I-A_\lambda^{-1}B)$ in the same way as the case of $\ker(I+A^{-1}B)$. 
Then in this case, we can state that the spectrum of $T_\ast$ includes $-1$, and 
the candidates of the eigenvalues are $1/(n+1)$ and $-1/(n+1)$. 
In any cases, it is easily check that the largest one is $\lambda_\ast$, which implies $\lambda_\ast$ is the largest eigenvalue of $T_\ast$. 
The corresponding eigenvector is obtained by (\ref{f}), (\ref{g}) and (\ref{h}). 
This completes the proof.

\subsection{Proof of Lemma \ref{IN}}
\label{section-IN}
Let $\tilde f_1 := f_1 / \|f_1\|$. 
Then
\begin{equation*}
\beta_-(a) = -\im \frac{\tilde f_1(t(a))}{\sqrt{\deg(t(a))}} = \frac{- \im }{\sqrt{n+1}} \frac{f_1(t(a))}{\|f_1\|},\quad \psi_{IN}(a) = \frac{1}{\sqrt{2(n+1)(n+2)}}.
\end{equation*}
Therefore
\begin{align*}
\langle \psi_{IN}, \beta_- \rangle &= \frac{- \im }{\sqrt{n+1}} \frac{1}{\sqrt{2(n+1)(n+2)} \|f_1\|} \sum_{a} f_1(t(a))\\
	&=  \frac{- \im }{\sqrt{n-1}} \frac{1}{\sqrt{2(n+1)(n+2)} \|f_1\|}\left( \sum_{a:t(a)\text{ has a self loop}} + \sum_{a : \text{otherwise}} \right) f_1(t(a))\\
	&=  \frac{-\im /\|f_1\| }{\sqrt{n+1}\sqrt{2(n+1)(n+2)} } \left\{ 4(n+1)\eta + 2n(n+1)\right\}\\
	&= \frac{-2\im }{\sqrt{2(n+2)}\|f_1\|}(2\eta + n).
\end{align*}
By characterizations of $\|f_1\|$ and $\eta$, we have $\|f_1\|\sim \sqrt{2n}\sim \sqrt{2(n+2)}$ and $\eta \sim 1$ as $n\to \infty$, and hence
\begin{equation*}
\frac{-2\im }{\sqrt{2(n+2)}\|f_1\|}(2\eta + n) \sim \frac{-\im }{n+2}(2+n) = -\im,
\end{equation*}
which completes the proof.

\subsection{Proof of Lemma \ref{Tar}}
\label{section-Tar}
By definition, we have
\begin{equation*}
\beta_+(a) = \frac{1}{|\sin \theta_1|}\frac{1}{\sqrt{n+1}}\times
	\begin{cases}
	\left\{\tilde f_1(o(a)) - \cos \theta_1 \tilde f_1(t(a))\right\} & \text{$a\not \in \text{(self loop)}$},\\
	(1+\cos \theta_1) \tilde f_1(o(a)) & \text{$a \in \text{(self loop)}$}.
\end{cases}
\end{equation*}
Now
\begin{align*}
\cos\theta_1 &= \frac{n-2 + \sqrt{n(n+8)}}{2(n+1)}\\
	&= \frac{1}{2(n+1)} \left\{n-2 + n\sqrt{1+\frac{8}{n}}\right\}\\
	&= \frac{1}{2\left(1+\frac{1}{n}\right)} \left\{1-\frac{2}{n} + \sqrt{ 1+\frac{8}{n} }\right\}.
\end{align*}
Letting $x = 2/n$ and 
\begin{equation*}
g(x) := \frac{1}{2+x}\left\{1- x + \sqrt{ 1+4x }\right\},
\end{equation*}
we have the asymptotic behavior of $g(x)$ near $x=0$ as follows:
\begin{equation*}
g(x) = 1-x^2 + O(x^3)\quad \text{ as }x\to 0,
\end{equation*}
which yields $\cos \theta_1 = 1+O(n^{-2})$ as $n\to \infty$.
Thus 
\begin{align*}
\sin \theta_1 &= \sqrt{1-\cos^2\theta_1}\\
	&\sim \sqrt{1-(1-x^2)^2} \sim \sqrt{(1-1+x^2)(1+1-x^2)}\\
	&\sim |x|\sqrt{2-x^2} \sim \sqrt{2}|x| = \frac{2\sqrt{2}}{n}.
\end{align*}
If $a$ is an endpoint of a self loop, we have
\begin{equation*}
\tilde f_1(o(a)) = \frac{1}{\|f_1\|}\eta \sim \frac{1}{\sqrt{2n}}
\end{equation*}
and hence
\begin{align*}
\langle \psi_{Tar}, \beta_+\rangle &= \frac{1}{2} \sum_{a\in \text{(self loop)}} \beta_+ (a)\\
	&\sim \frac{1}{2} \times \frac{1}{|\sin \theta_1|}\frac{1}{\sqrt{n+1}}(1+\cos \theta_1) \tilde f_1(o(a)) \times 4\\
	&\sim \frac{n}{2\sqrt{2}} \frac{2}{\sqrt{n+1}} \frac{1}{\sqrt{2n}} \times 2\\
	&= 1+o(1),
\end{align*}
which completes the proof.
\end{document}